\documentclass{article}
\usepackage[utf8]{inputenc}
\usepackage{setspace}
\usepackage[margin=1in]{geometry}
\usepackage{amssymb, bbm, amsthm}
\usepackage{booktabs}
\usepackage{publication}

\usepackage[table]{xcolor}
\usepackage{lmodern}
\usepackage{microtype}
\usepackage{mathtools}
\usepackage{array}
\usepackage{enumitem}
\usepackage{blkarray} 
\usepackage{tikz} 
\usepackage[hidelinks]{hyperref}
\usepackage[noabbrev,capitalize]{cleveref}
\usepackage[framemethod=TikZ]{mdframed}

\newtheorem{theorem}{Theorem}[section]
\newtheorem{lemma}[theorem]{Lemma}
\newtheorem{proposition}[theorem]{Proposition}
\newtheorem{corollary}[theorem]{Corollary}

\crefname{example}{example}{examples}
\newcounter{exampleBoxCounter}
\newcommand{\theExampleBox}{\arabic{exampleBoxCounter}}

\theoremstyle{definition}

\usepackage{etoolbox}
\theoremstyle{remark}

\AfterEndEnvironment{remark}{$\blacksquare$}

\newcounter{algorithm}
\renewcommand{\thealgorithm}{\arabic{algorithm}}

\begin{document}
\newcommand*\samethanks[1][\value{footnote}]{\footnotemark[#1]}

\title{A Practical Guide to Simulating Correlated Binary Outcomes}
\author[1]{Chi Heem Wong}
\contact{}
\affil[1]{Independent researcher. *Equal contribution.}
\author[1]{Zied Ben Chaouch}
\date{This version: Jul 17, 2026}
\maketitle

\begin{abstract}
Simulating dependent Bernoulli outcomes with prescribed means and pairwise Pearson correlations is a common task in risk modeling.
A familiar approach is the Gaussian-threshold workflow for binary outcomes, often viewed as a Bernoulli analogue of the Gaussian copula construction.
We show that setting latent Gaussian correlations equal to target Bernoulli correlations is generally incorrect after thresholding, and that pairwise tetrachoric calibration is exact only when the calibrated latent matrix is positive semidefinite.
We therefore formulate the problem directly over the joint Bernoulli probability mass function.
Given target means and pairwise correlations, we impose normalization, nonnegativity, mean constraints, and pairwise cross-moment constraints as a linear program over the $2^N$ atomic probabilities.
The resulting PMF formulation either returns an exact law matching the requested first and second moments or certifies infeasibility. A convex-hull characterization further shows that every feasible target admits a law supported on at most $1+N+\binom{N}{2}$ states, while every infeasible target admits a separating quadratic certificate.
We then develop a truncated-moment completion scheme that fits a reduced cross-moment table and generates samples by sequential conditioning, together with a sparse-support working-set refinement that can reduce memory usage on structured instances, although the worst-case complexity remains exponential.
Together, these constructions provide an exact PMF-based framework for feasibility and simulation at moderate dimension and structured alternatives when the full atomic representation is impractical, while clarifying the limits of Gaussian-threshold constructions.
\end{abstract}

\section{Introduction}
Correlated binary events are common in finance.
Company defaults, clinical drug development milestones, and option expirations can all be modeled as binary events, and modeling portfolios of such assets remains an active area of research.

In many applications, the input consists of a mean vector together with a target covariance matrix (or, equivalently, a correlation matrix and the means), and the task is to construct a simulator whose output matches those first two moments.
A familiar way to address this problem is to introduce a latent Gaussian model and then threshold it to obtain binary outcomes.
The latent-factor probit approach, widely used in credit risk, follows exactly this strategy: it matches the target Bernoulli means by applying probit thresholds to a low-rank latent Gaussian structure, so the induced dependence is constrained by the chosen factor specification.
A closely related full-rank construction replaces the low-rank factor model by a general latent correlation matrix.
We refer to this Bernoulli threshold construction as the ``Copula Method for Binary Outcomes,'' or CoMBO (see \Cref{sec:methods}), because it can be viewed as the binary adaptation of the Gaussian copula construction.

This latent-Gaussian route is attractive because it inherits familiar copula and factor-model machinery, but it does not directly resolve the underlying Bernoulli problem.
The naive version of CoMBO sets the latent Gaussian correlations equal to the target Bernoulli correlations before thresholding, but that rule is generally incorrect.
The natural pairwise correction is tetrachoric calibration, which yields a unique candidate latent correlation for each pair.
Even then, however, the resulting latent matrix need not be positive semidefinite: when it is, calibrated CoMBO is exact; when it is not, no exact Gaussian-threshold representation exists and any repaired implementation is necessarily approximate.
These observations suggest returning to the object that is actually sought, namely a joint law on $\{0,1\}^N$ with the prescribed first and second moments.

Our main technical contribution is therefore a direct formulation of this problem over the joint Bernoulli probability mass function.
We formulate an exact linear program over the $2^N$ atomic probabilities that matches the first and second moments and yields infeasibility certificates for moderate problem sizes. We characterize feasibility as membership in the convex hull of the binary first- and second-moment feature vectors; this yields both an $O(N^2)$ support bound for some compatible law and a separating quadratic witness whenever the target is infeasible.

This PMF linear program (PMF-LP) is the anchor of the paper: it answers the Bernoulli-feasibility question directly, without committing to any latent or parametric family.
We then develop two extensions that remain within the direct PMF and cross-moment framework.
First, we show how low-order cross-moment tables can be completed under a truncation assumption and then sampled by sequential conditioning, giving a reduced model class when a compact law is acceptable.
Second, we describe a sparse-support working-set refinement that seeks an exact law without materializing the full state space at once.

The rest of the paper is organized as follows.
We introduce the mathematical setup for correlated binary events in \Cref{sec:motivation}.
In \Cref{sec:methods}, we first review the Gaussian copula construction for continuous marginals and then explain the Bernoulli threshold adaptation used to generate correlated binary outcomes.
This background section clarifies why naive CoMBO does not, in general, produce the correct correlation between the generated variables and states the exact condition under which calibrated CoMBO works.
In \Cref{sec:bernoulli_pmf}, we then turn to the main contribution of the paper: we derive the linear moment constraints (normalization, nonnegativity, means, and pairwise cross-moments) that must hold for a joint Bernoulli distribution with the specified targets, and cast the task as a linear feasibility problem over the $2^N$ probabilities.
We first solve that exact problem with standard linear programming in \Cref{sec:linear_programming}, then introduce a truncated-moment completion with sequential conditioning in \Cref{sec:truncated_completion}.
We follow up with a short analysis of computational complexity in \Cref{sec:bernoulli_computation_complexity}.
The appendices collect a comparison of the main constructions, worked examples, detailed moment identities and truncation results, and the sparse-support refinement.
The accompanying code can be assessed at \url{https://github.com/chiheem/CorrelatedBinary}.

\section{Motivation}
\label{sec:motivation}
We consider a collection of $N$ binary outcomes $X_1, X_2, \dots, X_N$.
In the motivating portfolio interpretation, each $X_i$ represents the eventual success or failure of one program or asset, with
$P(X_i=1)=p_i$ and $P(X_i=0)=1-p_i$.
The user specifies the marginal means together with target pairwise Pearson correlations $\rho_{ij}$.
\Cref{eqn:first_moment,eqn:second_moment} summarize these first- and second-moment requirements.

\begin{equation}
\mathbb{E}[X_i] = p_i \qquad \forall i \in {1,\dots,N}
\label{eqn:first_moment}
\end{equation}
\begin{equation}
\operatorname{Corr}(X_i, X_j) = \rho_{ij} \qquad \forall i,j \in \{1,\dots,N\},\ i\neq j
\label{eqn:second_moment}
\end{equation}

For Bernoulli variables, each pair $(X_i,X_j)$ must also satisfy the following feasibility bounds \cite{chaganty2006}:
\begin{equation}
\rho_{\min} = \frac{\max\{0,\,p_i+p_j-1\} - p_i p_j}{\sqrt{p_i(1-p_i)\,p_j(1-p_j)}} \quad \le \; \rho_{ij} \; \le \quad
\rho_{\max} = \frac{\min\{p_i, p_j\} - p_i p_j}{\sqrt{p_i(1-p_i)\,p_j(1-p_j)}}.
\label{eqn:bernoulli_correlation_feasibility_bounds}
\end{equation}
These pairwise bounds are necessary for feasibility, but they are not sufficient to guarantee that a joint Bernoulli law exists for the full collection of requested moments.
The PMF-LP introduced later will serve as the exact global feasibility check.

To see why these inputs are natural in applications, consider a simple valuation model in which program $X_i$ pays $v_i^+$ upon success and $v_i^-$ upon failure.
Then the mean and variance of the portfolio value $V_p$ are:
\begin{align}
\mathbb{E}[V_p] 
&= \sum_{i=1}^N \big[ p_i \cdot v_i^+ + (1-p_i) \cdot v_i^- \big] \\
&= \sum_{i=1}^N v_i^- + \sum_{i=1}^N a_ip_i \hspace{2em}\text{where } a_i = v_i^+ - v_i^-
\label{eqn:portfolio_val_mean}
\end{align}
\begin{equation}
\operatorname{Var}(V_p)= \sum_{i=1}^N a_i^{2} \cdot p_i(1-p_i) +  2\sum_{1\le i<j\le N} a_i a_j \rho_{ij} \cdot \sqrt{p_i(1-p_i)\,p_j(1-p_j)}.
\label{eqn:portfolio_val_var}
\end{equation}

These formulas explain why means and pairwise correlations are natural primitives, but in many applications they are not the endpoint.
Milestone timing, early stopping, rebalancing, discounting, and downstream risk measures such as value-at-risk require Monte Carlo simulation rather than only closed-form moment calculations.
The practical task is therefore to generate Bernoulli vectors whose empirical first and second moments agree with a user-specified mean vector $\vec{\mu}$ and target correlation matrix $\mathbf{R}$.


\section{Gaussian-Threshold Methods for Binary Events}\label{sec:methods}
When target means and pairwise correlations are prescribed, two logically separate questions arise.
First, do those moments define any joint Bernoulli law at all?
Second, if they do, are they exactly representable by a latent-Gaussian threshold construction?
This section addresses only the second question.
We briefly review the Gaussian copula recipe for continuous marginals, define its Bernoulli threshold adaptation, and then state precisely when that adaptation is exact.
The first question, Bernoulli feasibility itself, is taken up in the next section through the PMF formulation.

\subsection{From the Gaussian copula construction to CoMBO}
A $d$-dimensional copula, $C:[0,1]^d \to [0, 1]$, is defined to be a cumulative distribution function (CDF) with uniform marginals.
By Sklar's theorem (1959) \cite{sklar1959fonctions}, for a $d$-dimensional cumulative distribution function (CDF) $F$ with marginals $F_1, \dots, F_d$, there exists a copula $C$ such that 
\begin{equation}
F(x_1, \dots, x_d)= C(F_1(x_1), \dots, F_d(x_d)) \qquad \forall x_i \in (-\infty, \infty), i \in \{1,2,\dots,d\}
\label{eqn:sklar_theorem}
\end{equation}

$C$ is unique if $F_i$ is continuous for all $i = 1, \dots, d$. Otherwise, $C$ is uniquely determined only on $Ran(F_1) \times \dots \times Ran(F_d)$, where $Ran(F_i)$ denotes the range of $F_i$. In the continuous case, $F_i(X_i) \sim U[0,1]$ and the copula $C$ is their joint distribution.

Since the probability integral transform is preserved under increasing transformations of each $X_i$, copulas are invariant under increasing transformations of the margins (i.e., they are `margin-free').
It has also been shown that any copula-based dependence measure, such as Spearman's correlation, is also margin-free if $F_1, \dots, F_d$ are continuous \cite{nelsen2006,joe2014}\footnote{Pearson correlation is \emph{not} margin-free: for a fixed copula, changing continuous marginals generally changes Pearson correlations. Thus a Gaussian copula with latent correlation $r$ preserves rank measures (e.g., Kendall's $\tau$, Spearman's $\rho_s$) but \emph{not} Pearson correlations. Matching Pearson correlations with non-Gaussian continuous marginals typically requires calibrating the latent $r$ (e.g., NORTA or Iman--Conover) \cite{cario1997norta,iman1982}.}.
The resulting separation between marginal behavior and dependence structure makes copulas a standard tool for modeling and simulating dependent variables with arbitrary continuous marginals.

For simplicity, the Gaussian copula is often chosen to model the dependence structure.
Sampling from a Gaussian copula with continuous marginals can be described in three steps:
\begin{enumerate}[itemsep=-1ex]
\item Sample random outcomes $\vec{z}$ from a multivariate Gaussian distribution, $N(0, \mathbf{R}_\text{lat})$. The latent correlation, $\mathbf{R}_\text{lat}$, is an internal parameter; for a target Pearson correlation matrix $\mathbf{R}$, it may need to be calibrated rather than simply set equal to $\mathbf{R}$. 
\item Apply the Gaussian CDF to each component of $\vec{z}$ to obtain a vector $\vec{u}$; $u_i \sim U[0,1]$ for all $i$.
\item[3.] (a) To obtain continuous variables, apply the inverse transformation $F_i^{-1}(u_i)$ along each component of $\vec{u}$.
\item[] (b) To obtain binary outcomes, prior literature applies an indicator function $\mathbbm{1}\{u_i < p_i\}$ (equivalently, $\mathbbm{1}\{z_i < \Phi^{-1}(p_i)\}$) to every component. Because this Bernoulli step is a thresholding adaptation rather than an invertible marginal transform, we refer to Steps 1, 2, and 3(b) together as the ``copula method for binary outcomes'', or \emph{CoMBO}.
\end{enumerate}

Step 3(a) is the classical Gaussian-copula recipe for continuous marginals.
Step 3(b) instead produces a latent-Gaussian threshold model with discrete Bernoulli margins.
Within this CoMBO family, \emph{naive CoMBO} sets $\mathbf{R}_\text{lat}=\mathbf{R}$, while \emph{pairwise-calibrated CoMBO} first solves for latent pairwise correlations before thresholding.
The distinction between Steps 3(a) and 3(b) is crucial: once the final map is a threshold rather than an invertible marginal transform, the latent correlation matrix becomes a calibration object, not the target Bernoulli correlation matrix itself.

\subsection{Why naive CoMBO fails, and when calibrated CoMBO is exact}\label{sec:wrong_correlation}
We now turn to the Gaussian-threshold representability question.
Suppose one wants to generate Bernoulli variables with means $p_i$ and target correlations $\rho_{ij}$ by thresholding a latent Gaussian vector.
The first issue is that the naive CoMBO does not preserve Pearson correlations.
The second is that even after correcting each pair individually, the calibrated pairwise latent correlations may fail to assemble into one valid multivariate Gaussian law.

\paragraph{Naive CoMBO generally misses the target Pearson correlations.}
Recall that CoMBO draws a latent Gaussian vector $Z\sim\mathcal{N}(0,\mathbf{R}_{\text{lat}})$, sets $U_i=\Phi(Z_i)$, and then produces Bernoulli outcomes by thresholding
$X_i=\mathbbm{1}\{U_i<p_i\}$, or equivalently, $X_i=\mathbbm{1}\{Z_i<\Phi^{-1}(p_i)\}$.
For continuous marginals, copulas are invariant under strictly increasing marginal transforms; however, thresholding is many-to-one and creates discrete marginals, for which copulas are not unique and the probability--integral transform is no longer uniform without randomization \cite{nelsen2006,joe2014,triverdi2005copula}.
As a result, Bernoulli Pearson correlations are nonlinear functions of the latent correlations and thresholds.

For a bivariate latent-Gaussian threshold model with latent correlation $r_{ij}$ and thresholds $t_i=\Phi^{-1}(p_i)$, $t_j=\Phi^{-1}(p_j)$, the induced Bernoulli Pearson correlation is
\begin{equation}
\operatorname{Corr}(X_i,X_j)
= \frac{\Phi_2(t_i,t_j;\, r_{ij}) - p_i p_j}{\sqrt{p_i(1-p_i)\,p_j(1-p_j)}},
\end{equation}
where $\Phi_2(\cdot,\cdot;\,r)$ is the bivariate standard normal CDF with correlation $r$.
Naively setting $r_{ij}=\rho_{ij}$ therefore does not, in general, yield $\operatorname{Corr}(X_i,X_j)=\rho_{ij}$.
This is the failure of naive CoMBO.

\paragraph{Pairwise calibration identifies the unique latent candidate, but multivariate consistency is still required.}
To target a desired Bernoulli Pearson correlation $\rho_{ij}$, one can instead \emph{calibrate} the latent correlation via the tetrachoric mapping \cite{emrich1991,qaqish2003}, i.e., solve
\begin{equation}\label{eqn:bernoulli_latent_calibration}
\Phi_2(t_i,t_j;\, r_{ij}) \;=\; p_i p_j \;+\; \rho_{ij}\,\sqrt{p_i(1-p_i)\,p_j(1-p_j)}
\quad\text{for } r_{ij}\in(-1,1).
\end{equation}
For continuous marginals, analogous calibration is standard in NORTA and Iman--Conover procedures \cite{cario1997norta,iman1982}.
Define the bivariate thresholding map
\[
g_{ij}(r)
:=
\frac{\Phi_2(t_i,t_j;\, r) - p_i p_j}{\sqrt{p_i(1-p_i)\,p_j(1-p_j)}}.
\]
Its derivative is
\[
g_{ij}'(r)
=
\frac{\phi_2(t_i,t_j;\, r)}{\sqrt{p_i(1-p_i)\,p_j(1-p_j)}} > 0,
\]
implying that $g_{ij}$ is strictly increasing and therefore invertible on its admissible range.
Thus pairwise calibration thus provides the \emph{unique} latent value
$
r_{ij}^\star = g_{ij}^{-1}(\rho_{ij})
$
required for exact CoMBO matching of pair $(i,j)$.
The remaining question is whether these calibrated pairwise values can be realized \emph{simultaneously} by one latent Gaussian vector.

Assemble the pairwise-calibrated values into the matrix
$
\mathbf{R}_{\mathrm{lat}}^\star := [r_{ij}^\star]
$;
If $\mathbf{R}_{\mathrm{lat}}^\star$ is a valid Gaussian correlation matrix (i.e., it has ones on the diagonal and is positive semidefinite) then calibrated CoMBO matches the requested Bernoulli means and pairwise correlations exactly.
Indeed, in this case, there exists $Z\sim \mathcal N(0,\mathbf{R}_{\mathrm{lat}}^\star)$, and every bivariate marginal $(Z_i,Z_j)$ has the calibrated correlation $r_{ij}^\star$, so thresholding reproduces $\rho_{ij}$ for every pair by construction.
On the other hand, if $\mathbf{R}_{\mathrm{lat}}^\star$ is not positive semidefinite (PSD), then no joint Gaussian latent vector exists with those pairwise calibrated entries.
Consequently, no exact Gaussian-threshold representation exists with the prescribed Bernoulli means and pairwise correlations.
Any practical repair that replaces $\mathbf{R}_{\mathrm{lat}}^\star$ by a nearby PSD matrix \emph{must} change at least one latent pair, and therefore changes at least one induced Bernoulli correlation as well.
A concrete numerical illustration, together with simulation results and the corresponding calculations, is given in \Cref{sec:simulation1}.

Note that the PSD test answers the Gaussian-threshold representability question, not the more basic Bernoulli-feasibility question.
The target moments may fail earlier because no joint Bernoulli law exists at all, in which case no exact simulator exists regardless of method.
Conversely, a target may be Bernoulli-feasible and yet still fall outside the Gaussian-threshold family because the unique pairwise-calibrated latent matrix is not PSD.
The next section therefore returns to the underlying Bernoulli problem and addresses feasibility directly through the joint PMF.

\section{The probability mass function (PMF) approach}\label{pmf_approach}
We now return to the first question: does there exist \emph{any} Bernoulli law with the requested first and second moments?
The most direct exact answer is to work with the joint PMF itself.
This direct PMF formulation serves two purposes;
First, it provides a global feasibility test when pairwise Bernoulli checks are inconclusive. Second, it yields an exact construction whenever the target is feasible but exact CoMBO is unavailable or undesirable.
We begin with the exact PMF linear program and then discuss several refinements and structured alternatives.
\Cref{table:method_comparison} summarizes the resulting constructions.

\subsection{The PMF of correlated Bernoulli}\label{sec:bernoulli_pmf}
Let $X=(X_1,\ldots,X_N)$ be a Bernoulli vector with $X_i\in\{0,1\}$.
Its joint probability mass function assigns a probability to each outcome $x=(x_1,\ldots,x_N)\in\{0,1\}^N$.

For bookkeeping, it is convenient to collect the $2^N$ possible outcomes into a matrix $M \in \{0,1\}^{2^N \times N}$:
\[
M = 
\begin{blockarray}{ccccc}
X_1 & X_2 & \cdots & X_{N-1} & X_N\\
\begin{block}{(ccccc)}
0 & 0 & \cdots & 0 & 0\\
0 & 0 & \cdots & 0 & 1\\
\vdots & \vdots & \ddots & \vdots & \vdots\\
1 & 1 & \cdots & 1 & 1\\
\end{block}
\end{blockarray}
\]

We represent the joint PMF as a vector $\vec{\alpha} \in [0,1]^{2^N}$ whose entries are the probabilities of the corresponding outcomes.
\[
\vec{\alpha} =
\begin{blockarray}{c}
\begin{block}{(c)}
p_{x_1=0, x_2=0, \ldots, x_N=0} \\
p_{x_1=0, x_2=0, \ldots, x_N=1} \\
\vdots\\
p_{x_1=1, x_2=1, \ldots, x_N=1} \\
\end{block}
\end{blockarray}
\]

The mean vector $\tilde{\mu}$ implied by $\vec{\alpha}$ is

\begin{equation}
    \tilde{\mu} = M^\mathsf{T}\vec{\alpha} \in \mathbb{R}^N
\label{equation:empirical_expectation}
\end{equation}

The covariance between variables $X_i$ and $X_j$, denoted $\tilde{\sigma}_{ij}$, is
\begin{equation}
\begin{aligned}
\tilde{\sigma}_{ij} 
&= \left( M[:,i] \circ M[:,j] \right)^\mathsf{T} \vec{\alpha} - (M[:,i]^\mathsf{T} \vec{\alpha}) \times (M[:,j]^\mathsf{T} \vec{\alpha}) \\
&= \vec{\kappa}_{ij}^\mathsf{T} \vec{\alpha} - \tilde{\mu}_i \tilde{\mu}_j \quad \textrm{where }\vec{\kappa}_{ij} = M[:,i] \circ M[:,j]
\end{aligned}
\end{equation}

These identities are the key point for the exact PMF-LP: the target means and pairwise cross-moments become linear constraints in the atomic probabilities $\vec{\alpha}$.

\subsection{Linear programming approach for obtaining the PMF (PMF-LP)}\label{section:opt_prog_maths}\label{solving_pmf}
Suppose the user provides the target mean vector $\vec{\mu}$ and target correlation matrix $\mathbf{R}$.
The corresponding covariance matrix $\mathbf{\Sigma}$ is

\begin{equation}
\sigma_{ij} =
\begin{cases}
\mu_i(1-\mu_i), & i=j, \\
\rho_{ij}\,\sqrt{\mu_i(1-\mu_i)\,\mu_j(1-\mu_j)}, & i \neq j,
\end{cases}
\end{equation}
where $\sigma_{ij}$ and $\rho_{ij}$ are elements in the covariance and correlation matrices, respectively.

Our goal is to recover a feasible $\vec{\alpha}$ such that the induced means match $\vec{\mu}$, the induced covariances match $\mathbf{\Sigma}$, and the probability axioms are respected.

\paragraph{Feasibility program}\label{sec:linear_programming}
This leads to the following feasibility program:
\begin{subequations}\label{eqn:linear_prog}
\begin{align}
\min_{\vec{\alpha}} 0 & \quad\quad \textrm{subject to:}\\
\mathbf{1}^\mathsf{T} \vec{\alpha} &= 1\\
\vec{\alpha} &\geq \vec{0} \\
M^\mathsf{T} \vec{\alpha} &= \vec{\mu} \\
\vec{\kappa}_{ij}^\mathsf{T} \vec{\alpha} &= \mathbf{\Sigma}_{ij} + \mu_i\mu_j \quad \forall 1 \le i < j \le N
\end{align}
\end{subequations}

This is a pure feasibility LP.
Any standard solver can be used to search for a compatible $\vec{\alpha}$, and a feasible solution immediately provides an exact simulator through sampling from the recovered PMF.
\Cref{sec:pmf_lp_refinement} discusses a sparse-support refinement that works with a much smaller working set of states when the full state space is too large to materialize.

In general, $\vec{\alpha}$ is not unique.
That nonuniqueness is acceptable here: the role of the PMF-LP is to certify feasibility and produce one compatible law, not to identify a canonical one.
The geometric interpretation and proofs are deferred to \Cref{prop:moment_polytope} in the appendix.
In particular, every feasible target admits a compatible distribution supported on at most
$1+N+\binom{N}{2}$ states, while every infeasible target admits a separating quadratic certificate.
These facts explain both why sparse solutions exist and how LP infeasibility can be witnessed in terms of
the requested first and second moments.

\paragraph{Sampling after solving the LP.}
\label{sec:pmf_lp_algorithm}
Once a feasible $\vec{\alpha}$ has been obtained, correlated Bernoulli draws can be generated by inverse transform sampling over the $2^N$ atoms.

\paragraph{Optional tightening and pre-checks.}
Before solving, one should first screen for obvious infeasibility by checking that each target pair $(\mu_i,\mu_j,\rho_{ij})$ satisfies the Bernoulli bounds in \Cref{eqn:bernoulli_correlation_feasibility_bounds}.
These pairwise checks are only necessary conditions; the LP itself is the global feasibility test.
If one allows tolerances or ranges for pairwise moments instead of hard equalities, the following linear cuts can be added to enforce bounds on $\mathbb{E}[X_iX_j]$:
\[
\max\{0,\,\mu_i+\mu_j-1\} \;\le\; \vec{\kappa}_{ij}^\mathsf{T} \vec{\alpha} \;\le\; \min\{\mu_i,\mu_j\}\quad (1\le i<j\le N),
\]
In the exact-matching PMF-LP these inequalities are redundant, but they remain useful as cheap pre-checks and as explicit cuts in relaxed or robust variants.

\paragraph{Computational complexity.}\label{sec:bernoulli_computation_complexity}
Given $\vec{\alpha}$, drawing random outcomes with the inverse transform sampling method is straightforward.
Accordingly, the computational complexity of the approach is dominated by setting up the LP and solving it.
The LP has $2^N$ decision variables (the entries of $\vec{\alpha}$), $1+N+\binom{N}{2}$ equality constraints, and $2^N$ nonnegativity bounds.
Interior-point methods solve LPs in time polynomial in the LP input size \cite{cohen2021solving}, but the LP input itself grows exponentially in $N$ because of the $2^N$ states.
Therefore, this formulation provides exact moment matching and feasibility certificates for moderate $N$, but it is not polynomial-time in $N$.

\Cref{appendix:example} illustrates the three-variable case explicitly. \Cref{appendix:moment_identities} contains the full PMF/CDF identities, proofs, and further extensions.

\subsection{Truncated moment completion and sequential conditioning}\label{sec:truncated_completion}
The exact PMF-LP is the right tool when one wants a definitive feasibility certificate over the full $2^N$ state space.
In some applications, however, the goal is instead to construct a compatible law from low-order information without materializing every atom.
This leads naturally to a reduced cross-moment representation.

For this subsection it is convenient to index atoms and moments by subsets of $[N]:=\{1,\dots,N\}$.
For each $A\subseteq[N]$, let
\[
p_A := \mathbb{P}(X_i=1 \text{ iff } i\in A),
\qquad
\gamma_A := \mathbb{E}\!\left[\prod_{i\in A} X_i\right],
\qquad
\gamma_\varnothing := 1.
\]
Here $p_A$ is the probability of the atom whose set of ones is exactly $A$, while $\gamma_A$ is the corresponding cross-moment.
The M\"obius-inversion identities linking $(p_A)$ and $(\gamma_A)$ are collected in \Cref{appendix:moment_identities}; below we use only the formulas needed for truncated completion.

Fix a truncation order $k\ge 2$.
We keep
\[
\gamma_\varnothing = 1,\qquad
\gamma_{\{i\}} = \mu_i,\qquad
\gamma_{\{i,j\}} = \mu_i\mu_j + \rho_{ij}\sqrt{\mu_i(1-\mu_i)\,\mu_j(1-\mu_j)},
\]
treat the moments $\gamma_A$ with $3\le |A|\le k$ as decision variables, and set $\gamma_A=0$ for $|A|>k$.
The reduced completion problem therefore replaces the $2^N$ atomic variables of the exact PMF-LP with only
$
\sum_{m=3}^{k}\binom{N}{m}
$
higher-order moment variables.
By the M\"obius inversion formulas in \Cref{appendix:moment_identities}, the implied atomic probabilities satisfy
\begin{equation}\label{eq:main_truncated_atom}
p_A
=
\sum_{\substack{B\supseteq A\\ |B|\le k}}
(-1)^{|B|-|A|}\gamma_B
\qquad (|A|\le k),
\end{equation}
and $p_A=0$ for $|A|>k$.
Hence one can solve the reduced feasibility problem
\[
\text{find } (\gamma_A)_{A\subseteq[N],\,3\le |A|\le k}
\]
subject to the nonnegativity constraints
\[
p_A \ge 0
\qquad \text{for all } A\subseteq[N]\text{ with }|A|\le k,
\]
where $p_A$ is defined by \eqref{eq:main_truncated_atom}.
This should be interpreted as a model class rather than a universal feasibility test: infeasibility may mean either that the target moments are globally impossible or that the chosen truncation order is too small.
The most aggressive case $k=2$ gives an explicit pairwise-only construction supported on outcomes with at most two ones; \Cref{appendix:moment_identities} records the closed forms and the associated limitations.

Once a reduced table $\Gamma_k=\{\gamma_A:|A|\le k\}$ has been fitted, one can sample without materializing the full PMF.
For disjoint sets $O,Z\subseteq[N]$, define the cylinder probability
\[
w(O,Z)
:=
\mathbb{P}(X_i=1\ \forall i\in O,\ X_j=0\ \forall j\in Z)
=
\sum_{\substack{B\subseteq Z\\ |O|+|B|\le k}}
(-1)^{|B|}\gamma_{O\cup B}.
\]
If the variables are sampled in an order $\pi(1),\dots,\pi(N)$ and $O_m,Z_m$ denote the realized ones and zeros after $m$ steps, then
\begin{equation}\label{eq:main_truncated_conditional}
\mathbb{P}\!\big(X_{\pi(m+1)}=1 \,\big|\, X_i=1\ \forall i\in O_m,\ X_j=0\ \forall j\in Z_m\big)
=
\frac{w(O_m\cup\{\pi(m+1)\},Z_m)}{w(O_m,Z_m)},
\end{equation}
whenever $w(O_m,Z_m)>0$.
For fixed $k$, each evaluation involves at most
$
M_k := \sum_{r=0}^{k}\binom{N}{r}
$
terms, so this sequential-conditioning sampler is polynomial in $N$ for fixed truncation order, though only within the truncated family.

\Cref{alg:main_truncated_sampler} summarizes the workflow.
\Cref{appendix:moment_identities} contains the full derivations, sharper complexity statements, and impossibility results showing that small $k$ cannot be universally sufficient.

\begin{figure}[h]
\centering
\refstepcounter{algorithm}\label{alg:main_truncated_sampler}
\fbox{%
\begin{minipage}{0.95\linewidth}
\small
\textbf{Algorithm~\thealgorithm. Order-$k$ truncated moment completion with sequential conditioning}

\textbf{Input:} target means $(\mu_i)_{i=1}^N$, target pairwise correlations $(\rho_{ij})_{1\le i<j\le N}$, truncation order $k$, number of samples $L$, and an ordering $\pi(1),\dots,\pi(N)$.

\begin{enumerate}[itemsep=0.25ex, topsep=0.4ex, leftmargin=2.4em]
\item Set $\gamma_\varnothing=1$, $\gamma_{\{i\}}=\mu_i$, and
\[
\gamma_{\{i,j\}}=\mu_i\mu_j+\rho_{ij}\sqrt{\mu_i(1-\mu_i)\,\mu_j(1-\mu_j)}
\qquad (1\le i<j\le N).
\]
\item Introduce $\gamma_A$ for all $A\subseteq[N]$ with $3\le |A|\le k$, set $\gamma_A=0$ for $|A|>k$, and solve
\[
\sum_{\substack{B\supseteq A\\ |B|\le k}}(-1)^{|B|-|A|}\gamma_B \ge 0
\qquad \text{for all } A\subseteq[N]\text{ with }|A|\le k.
\]
\item If the LP is infeasible, increase $k$ or abandon the truncated family. Otherwise store the fitted table $\Gamma_k$.
\item For each sample, initialize $O_0=Z_0=\varnothing$ and process variables in the order $\pi(1),\dots,\pi(N)$.
\item At step $m$, compute
\[
q_m=
\frac{
\sum_{\substack{B\subseteq Z_m\\ |O_m|+1+|B|\le k}}(-1)^{|B|}\gamma_{O_m\cup\{\pi(m+1)\}\cup B}
}{
\sum_{\substack{B\subseteq Z_m\\ |O_m|+|B|\le k}}(-1)^{|B|}\gamma_{O_m\cup B}
},
\]
draw $U_m\sim\operatorname{Unif}(0,1)$, and set $X_{\pi(m+1)}=1$ iff $U_m\le q_m$.
\item Update $O_m$ and $Z_m$ and continue until all coordinates are assigned.
\end{enumerate}

\textbf{Output:} samples from the fitted order-$k$ truncated model.
\end{minipage}%
}
\end{figure}

\subsection{Code}
We provide the code for PMF-LP and truncated moment completion with sequential conditioning at \url{https://github.com/chiheem/CorrelatedBinary}.
The repository also contains the code used to generate the simulations in \Cref{sec:simulation1}.

\section{Related Literature}
The literature most relevant to this paper can be grouped into four strands.

The first strand concerns representation formulas for multivariate Bernoulli laws.
Bahadur \cite{bahadur1961} expresses a multivariate Bernoulli law in terms of its marginals and higher-order interaction terms, while Teugels \cite{teugels1990} gives vectorized and Kronecker-product representations for multivariate Bernoulli and binomial laws.
More recently, Fontana and Semeraro \cite{fontana2018} study classes of multivariate Bernoulli distributions with specified moments and show how such classes can be characterized and sampled algorithmically.
Our appendix formulas belong to this representation tradition, but our emphasis is more operational and more narrowly focused: given low-order target moments, we want one compatible law or a certificate that none exists, with the PMF-LP as the central exact construction.

The second strand studies simulation procedures within specific model families.
Emrich and Piedmonte \cite{emrich1991} use a multivariate-normal thresholding construction, while Qaqish \cite{qaqish2003} proposes a sequential family of multivariate binary laws tailored to simulation with specified marginal means and correlations.
These methods are practically important, but they answer a narrower question than the one considered here.
They ask whether a target can be fitted inside a particular latent or conditional family.
Our PMF-LP is meant to answer the broader question of whether the target moments are feasible at all, independently of any chosen family.

The third strand concerns feasibility and extremal questions under partial moment information.
Chaganty and Joe \cite{chaganty2006} study admissible correlation ranges for Bernoulli vectors, showing that even pairwise feasibility is highly constrained.
Padmanabhan and Natarajan \cite{PadmanabhanNatarajan2021} derive sharp bounds on tail probabilities of $\sum_i X_i$ from bivariate Bernoulli marginals on a graph.
Their objective differs from ours: they bound aggregate-event probabilities, whereas we seek an explicit joint law on $\{0,1\}^N$ that can actually be sampled once the first two moments are prescribed.

The fourth strand consists of exchangeable or count-based models, which are especially relevant in credit applications.
Kadane \cite{kadane2016sums} studies the distribution of the sum itself and proposes the Conway--Maxwell--binomial family as a flexible count model that accommodates both positive and negative association.
These models are useful benchmarks for homogeneous portfolios and overdispersed count data, but they are too restrictive for our setting because our input is a full, generally non-exchangeable matrix of target moments and our goal is a multivariate simulator rather than only a law for the sum.

\section{Discussion}
This paper addresses a practical problem: how should one simulate Bernoulli vectors with prescribed means and pairwise correlations?
Our main contributions are a direct PMF formulation of this task and a linear program that decides whether the requested moments are jointly feasible and, when they are, produces an exact law that matches them.

CoMBO, while familiar to many readers, addresses only the subset of targets that admit a latent-Gaussian threshold representation.
The naive CoMBO generally misses the target Pearson correlations, and pairwise tetrachoric calibration is exact only when the calibrated latent matrix is a valid Gaussian correlation matrix.
A Bernoulli target may therefore be feasible while still failing to admit an exact Gaussian-threshold representation.
The numerical illustration highlights this multivariate consistency obstruction rather than any failure of bivariate calibration.

By contrast, the PMF-LP is exact and transparent.
It separates the feasibility question from the choice of model family and avoids imposing latent or parametric structure in advance.
Its main limitation is computational: because the state space has size $2^N$, exact global moment matching remains expensive in high dimensions even when the formulation itself is simple.

Two constructions address that bottleneck while remaining within the direct Bernoulli framework.
The truncated completion method works with cross-moments under an explicit low-order assumption, reducing the representation at the cost of restricting the admissible support.
The sparse-support refinement pursues the exact PMF-LP objective on a small working set of states.
The convex-hull result guarantees that some exact support has size at most $1+N+\binom{N}{2}$ whenever the target is feasible, although finding that support offers no worst-case escape from exponential complexity.

Several directions for future work appear especially natural.
First, it would be valuable to strengthen the exact formulations through symmetry reduction, sparse interaction structure, support reduction, cut generation, and column-generation ideas that avoid materializing all $2^N$ variables at once.
Second, it would also be useful to develop controlled relaxations that enforce critical moments exactly while allowing bounded error on others, together with adaptive rules for selecting the truncation order from the target data.
Third, the relationship between truncated and sparse-support representations deserves a more systematic empirical study, since both aim to compress the search space but do so through rather different structural assumptions.

Overall, exact global moment matching for high-dimensional Bernoulli vectors remains difficult.
This paper clarifies three separate issues: feasibility of the requested moments, exact representability within the Gaussian-threshold family, and the computational cost of constructing an exact law.
That distinction helps identify when familiar approximations are appropriate and where further algorithmic progress is most needed.

\vspace{24pt}
\subsection*{Disclaimers}
\paragraph{Independent work}
This is independent work by the authors, performed outside of their full-time jobs, and no support of any kind was provided by any organization. The views and opinions expressed in this work are solely those of the authors and do not represent the views or opinions of their employers.

\paragraph{Use of AI.}
Large language models, under the supervision of the authors, have been used to refine the article for clarity, check for logical consistency, and tighten the proofs.

\subsection*{Acknowledgements}
Chi Heem and Zied thank Andrew Ang for his valuable comments and feedback on the paper.

\clearpage
\bibliographystyle{plain}
\bibliography{references}

@article{cohen2021solving,
  title={Solving linear programs in the current matrix multiplication time},
  author={Cohen, Michael B and Lee, Yin Tat and Song, Zhao},
  journal={Journal of the ACM (JACM)},
  volume={68},
  number={1},
  pages={1--39},
  year={2021},
  publisher={ACM New York, NY, USA}
}

@inproceedings{sklar1959fonctions,
  title={Fonctions de r{\'e}partition {\`a} n dimensions et leurs marges},
  author={Sklar, M},
  booktitle={Annales de l'ISUP},
  volume={8},
  number={3},
  pages={229--231},
  year={1959}
}

@book{nelsen2006,
  author = {Nelsen, Roger B.},
  title = {An Introduction to Copulas},
  publisher = {Springer},
  year = {2006}
}

@book{joe2014,
  author = {Joe, Harry},
  title = {Dependence Modeling with Copulas},
  publisher = {CRC Press},
  year = {2014}
}

@article{triverdi2005copula,
  title={Copula modeling: An introduction for practitioners},
  author={Triverdi, PK and Zimmer, David M},
  journal={Foundations and Trends in Econometrics},
  volume={1},
  number={1},
  pages={1--111},
  year={2005}
}

@article{chaganty2006,
  author = {Chaganty, Narasimhan R. and Joe, Harry},
  title = {Range of correlation matrices for dependent {B}ernoulli random variables},
  journal = {Biometrika},
  year = {2006}
}

@article{emrich1991,
  author = {Emrich, Laurie J. and Piedmonte, Marion R.},
  title = {A method for generating high-dimensional correlated binary variates with specified marginal means and correlations},
  journal = {The American Statistician},
  volume = {45},
  number = {4},
  pages = {302--304},
  year = {1991}
}

@article{qaqish2003,
  author = {Qaqish, Bahjat F.},
  title = {A family of multivariate binary distributions for simulating correlated binary variables with specified marginal means and correlations},
  journal = {Biometrika},
  volume = {90},
  number = {2},
  pages = {455--463},
  year = {2003},
  doi = {10.1093/biomet/90.2.455}
}

@article{bahadur1961,
  author = {Bahadur, Raghu Raj},
  title = {A representation of the joint distribution of responses to $n$ dichotomous items},
  journal = {The Annals of Mathematical Statistics},
  year = {1961}
}

@techreport{cario1997norta,
  author = {Cario, Celso M. and Nelson, Barry L.},
  title = {Modeling and Generating Random Vectors with Arbitrary Marginal Distributions and Correlation Matrix},
  institution = {Department of Industrial Engineering and Management Sciences, Northwestern University},
  year = {1997},
  note = {Technical Report}
}

@article{iman1982,
  author = {Iman, Ronald L. and Conover, William J.},
  title = {A Distribution-Free Approach to Inducing Rank Correlation Among Input Variables},
  journal = {Communications in Statistics - Simulation and Computation},
  volume = {11},
  number = {3},
  pages = {311--334},
  year = {1982}
}

@article{Vorobev1962,
  author  = {Nikolai N., Vorob'ev},
  title   = {Consistent Families of Measures and Their Extensions},
  journal = {Theory of Probability and its Applications},
  volume  = {7},
  number  = {2},
  pages   = {147--163},
  year    = {1962},
  doi     = {10.1137/1107014}
}

@article{Hoeffding1963,
  author  = {Wassily Hoeffding},
  title   = {Probability Inequalities for Sums of Bounded Random Variables},
  journal = {Journal of the American Statistical Association},
  volume  = {58},
  number  = {301},
  pages   = {13--30},
  year    = {1963},
  doi     = {10.1080/01621459.1963.10500830}
}

@article{PadmanabhanNatarajan2021,
  author  = {Divya Padmanabhan and Karthik Natarajan},
  title   = {Tree Bounds for Sums of {B}ernoulli Random Variables: A Linear Optimization Approach},
  journal = {INFORMS Journal on Optimization},
  volume  = {3},
  number  = {1},
  pages   = {23--45},
  year    = {2021}
}

@article{teugels1990,
  author  = {Teugels, Jozef L.},
  title   = {Some representations of the multivariate {B}ernoulli and binomial distributions},
  journal = {Journal of Multivariate Analysis},
  volume  = {32},
  number  = {2},
  pages   = {256--268},
  year    = {1990},
  doi     = {10.1016/0047-259X(90)90084-U}
}

@article{fontana2018,
  author  = {Fontana, Roberto and Semeraro, Patrizia},
  title   = {Representation of multivariate {B}ernoulli distributions with a given set of specified moments},
  journal = {Journal of Multivariate Analysis},
  volume  = {168},
  pages   = {290--303},
  year    = {2018},
  doi     = {10.1016/j.jmva.2018.08.003}
}

@article{kadane2016sums,
  author = {Kadane, Joseph B.},
  title = {Sums of Possibly Associated {B}ernoulli Variables: The {Conway--Maxwell--Binomial} Distribution},
  journal = {Bayesian Analysis},
  volume = {11},
  number = {2},
  pages = {403--420},
  year = {2016},
  doi = {10.1214/15-BA955}
}

\clearpage
\appendix
\renewcommand{\thesection}{\Alph{section}}
\crefalias{section}{appendix}
\section{Comparison of the Main Constructions}\label{appendix:comparison}
\Cref{table:method_comparison} summarizes the main constructions discussed in the paper and highlights the trade-offs among exactness, representation, and computational cost.
The point is not that one method dominates the others, but that they answer slightly different versions of the same modeling problem.

\begin{table}[!h]
\centering
\small
\setlength{\tabcolsep}{4pt}
\renewcommand{\arraystretch}{1.18}
\caption{Comparison of the main constructions discussed in the paper.}
\label{table:method_comparison}
\resizebox{\textwidth}{!}{
\begin{tabular}{p{0.15\linewidth} p{0.2\linewidth} p{0.18\linewidth} p{0.22\linewidth} p{0.23\linewidth}}
\toprule
\textbf{Construction} & \textbf{Representation} & \textbf{Exact wrt target moments?} & \textbf{Main computational profile} & \textbf{Most useful when / main caveat} \\
\midrule

CoMBO (naive or pairwise-calibrated) & Latent Gaussian law plus Bernoulli thresholding at probit cutoffs & Yes when the tetrachorically calibrated latent matrix is PSD; otherwise no & Pairwise root-finding plus a PSD check on the assembled latent matrix; any PSD repair sacrifices exactness & Useful when an exact Gaussian-threshold representation exists or as a familiar baseline. Main caveat: many feasible Bernoulli targets are not exactly representable this way. \\

\hline
PMF-LP & Full atomic PMF $(p_x)_{x\in\{0,1\}^N}$ or $\vec{\alpha}\in\mathbb{R}^{2^N}$ & Yes, when feasible & LP over $2^N$ variables with $1+N+\binom{N}{2}$ equality constraints & Best when exact feasibility certificates matter and $N$ is moderate. Main caveat: exponential state space. \\

Truncated completion + sequential conditioning & Cross-moment table $(\gamma_A)$ with $\gamma_A=0$ for $|A|>k$, followed by conditional sampling & Exact only within the chosen truncated family & Reduced LP over moments up to order $k$, then sampling via conditional ratios computed from the fitted table & Useful when a compact low-order representation is acceptable. Main caveat: infeasibility may reflect an overly small truncation order rather than true global impossibility. \\

Sparse-support refinement & Working set of candidate states with weights updated by restricted master LPs and pricing heuristics & Potentially yes, if the working set reaches an exact fit & Avoids storing all $2^N$ states at once, but may still require many iterations or hard pricing steps & Useful when the true or near-feasible law is expected to have small support. Main caveat: no worst-case escape from exponential behavior. \\
\bottomrule
\end{tabular}
}
\end{table}

\newpage
\section{A Simple CoMBO Illustration}\label{sec:simulation1}
This appendix records a simple three-variable experiment illustrating two points from the main text: naive CoMBO can miss the target correlations badly, and pairwise calibration alone does not guarantee an exact multivariate CoMBO generator.
We attempt to simulate three Bernoulli variables $X_1, X_2, X_3$ with the following means and correlations:

\begin{equation*}
\vec{\mu} = \left( 
\begin{array}{ccc}
0.2 & 0.7 & 0.6 \\
\end{array} \right) ^\mathsf{T}
\end{equation*}

\begin{equation*}
\mathbf{R} = \left( 
\begin{array}{ccc}
1   & 0.1 & 0.4\\
0.1 & 1   & 0.8\\
0.4 & 0.8 & 1
\end{array} \right)
\end{equation*}

In each replication, we generate 10,000 realizations using naive CoMBO and estimate the sample means and correlation matrix.
We repeat the procedure 1,000 times and report the mean and standard deviation of those estimates.
We then repeat the experiment using pairwise-calibrated CoMBO.
The point of the calibrated experiment is not that tetrachoric calibration is wrong in the bivariate sense; rather, it is to illustrate that pairwise calibration alone need not yield an exact realizable multivariate Bernoulli generator.
For comparison, we also generate continuous variables with logistic marginals, both with and without latent-correlation calibration.
The calibrated latent matrices are obtained by solving \Cref{eqn:bernoulli_latent_calibration} pairwise for the binary threshold model, and by solving the analogous pairwise root-finding problem for logistic marginals, where each transformed Pearson correlation is evaluated deterministically by Gauss--Hermite quadrature.

The results are shown in \Cref{table:simulation_results}.
Setting $\mathbf{R}_\text{lat} = \mathbf{R}$ leads to visible correlation error, and the effect is substantially larger for the Bernoulli threshold model than for the continuous logistic benchmark.
The calibrated Bernoulli experiment still does not recover the target multivariate correlation structure exactly, whereas the calibrated continuous experiment does.

A simple calculation shows that the target is pairwise Bernoulli-feasible, which we can verify using \Cref{eqn:bernoulli_correlation_feasibility_bounds}:
\[
\rho_{12}=0.1\in[-0.764,\;0.327],\qquad
\rho_{13}=0.4\in[-0.612,\;0.408],\qquad
\rho_{23}=0.8\in[-0.535,\;0.802].
\]

Solving the pairwise tetrachoric calibration equations in \Cref{eqn:bernoulli_latent_calibration} gives:
\[
\mathbf{R}_{\mathrm{lat}}^\star \approx
\begin{pmatrix}
1 & 0.1968 & 0.8557\\
0.1968 & 1 & 0.9910\\
0.8557 & 0.9910 & 1
\end{pmatrix}.
\]
For a $3\times 3$ correlation matrix with unit diagonal, positive semidefiniteness requires
\[
\det(\mathbf{R}_{\mathrm{lat}}^\star)
=
1+2r_{12}^\star r_{13}^\star r_{23}^\star-(r_{12}^\star)^2-(r_{13}^\star)^2-(r_{23}^\star)^2
\ge 0.
\]
Here,
\[
\det(\mathbf{R}_{\mathrm{lat}}^\star)\approx -0.41924 < 0,
\]
so the calibrated latent matrix is not PSD.
This explains the calibrated Bernoulli columns in \Cref{table:simulation_results}: the target is pairwise Bernoulli-feasible, but the uniquely calibrated latent pairs do not assemble into one valid multivariate Gaussian correlation matrix.
By contrast, naive CoMBO fails earlier because thresholding changes the Pearson correlations even before this multivariate PSD check is considered.

\begin{table}
\caption{
Summary of 1,000 iterations of 10,000 simulated realizations.
Setting $\mathbf{R}_\text{lat} = \mathbf{R}$ produces incorrect results, with larger errors for the Bernoulli threshold model than for the continuous logistic benchmark.
After calibration, the generated continuous variables recover the desired correlations whereas the generated discrete variables do not.
}
\label{table:simulation_results}
\centering
\vspace{0.75em}
\begin{tabular}{c | c | c c | c c | c c | c c }
\toprule
{} & {} & \multicolumn{4}{c|}{Bernoulli}  & \multicolumn{4}{c}{Logistic}\\
{} & {} & \multicolumn{2}{c|}{Uncalibrated} & \multicolumn{2}{c|}{Calibrated} & \multicolumn{2}{c|}{Uncalibrated} & \multicolumn{2}{c}{Calibrated}\\
\cline{3-10}
{} & {Target} & Mean & S.D. & Mean & S.D. & Mean & S.D. & Mean & S.D.\\
\cline{1-10}
$p_1$ & 0.200 & 0.200 & 0.004 & 0.200 & 0.004 & 0.199 & 0.018 & 0.200 & 0.018 \\
$p_2$ & 0.700 & 0.700 & 0.005 & 0.700 & 0.004 & 0.700 & 0.018 & 0.700 & 0.018 \\
$p_3$ & 0.600 & 0.600 & 0.005 & 0.600 & 0.005 & 0.599 & 0.018 & 0.600 & 0.018 \\
\hline
$\rho_{12}$ & 0.100 & 0.052 & 0.009 & 0.135 & 0.009 & 0.099 & 0.010 & 0.100 & 0.010 \\
$\rho_{13}$ & 0.400 & 0.212 & 0.009 & 0.367 & 0.006 & 0.397 & 0.008 & 0.400 & 0.008 \\
$\rho_{23}$ & 0.800 & 0.567 & 0.008 & 0.617 & 0.008 & 0.798 & 0.004 & 0.800 & 0.004 \\
\bottomrule
\end{tabular}
\end{table}

\clearpage
\section{A Three-Variable PMF Illustration}
\label{appendix:example}
We derive the constraints explicitly for three correlated Bernoulli variables in order to make the PMF construction concrete.
The joint distribution table is shown in \Cref{table:example_joint_dist}.

\begin{table}[h!]
\centering
\begin{tabular}{c | c c c | c}
\toprule
Index & $X_1$ & $X_2$ & $X_3$ & Probability \\
\midrule
0 & 0 & 0 & 0 & $\alpha_0$ \\
1 & 0 & 0 & 1 & $\alpha_1$ \\
2 & 0 & 1 & 0 & $\alpha_2$ \\
3 & 0 & 1 & 1 & $\alpha_3$ \\
4 & 1 & 0 & 0 & $\alpha_4$ \\
5 & 1 & 0 & 1 & $\alpha_5$ \\
6 & 1 & 1 & 0 & $\alpha_6$ \\
7 & 1 & 1 & 1 & $\alpha_7 = 1-\sum_{l=0}^{6}\alpha_l$ \\
\bottomrule
\end{tabular}
\caption{The probability mass function for the case of three binary variables $X_1, X_2, X_3$.}
\label{table:example_joint_dist}
\end{table}

The first moment conditions are given by \Cref{equation:first_moment_example}.
\begin{subequations}
\begin{align}
\mathbb{E}[X_1] = \alpha_4 + \alpha_5 + \alpha_6 + \alpha_7\\
\mathbb{E}[X_2] = \alpha_2 + \alpha_3 + \alpha_6 + \alpha_7\\
\mathbb{E}[X_3] = \alpha_1 + \alpha_3 + \alpha_5 + \alpha_7
\end{align}
\label{equation:first_moment_example}
\end{subequations}

The expectations of the products between pairs, $\mathbb{E}[X_iX_j]$, are given in \Cref{table:example_crossmoments}.
\begin{table}[!h]
\centering
\begin{tabular}{c|c c c}
\toprule
{}    & $X_1$             & $X_2$         & $X_3$ \\
\midrule
$X_1$ & $\alpha_4+\alpha_5+\alpha_6+\alpha_7$ & $\alpha_6+\alpha_7$               & $\alpha_5+\alpha_7$ \\
$X_2$ & $\alpha_6+\alpha_7$         & $\alpha_2 + \alpha_3 + \alpha_6 + \alpha_7$ & $\alpha_3+\alpha_7$ \\
$X_3$ & $\alpha_5+\alpha_7$         & $\alpha_3+\alpha_7$               & $\alpha_1 + \alpha_3 + \alpha_5 + \alpha_7$ \\
\bottomrule
\end{tabular}
\caption{$\mathbb{E}[X_iX_j]$ for various pairs in the case of three binary variables $X_1, X_2, X_3$.}
\label{table:example_crossmoments}
\end{table}

The covariance and correlation can then be written explicitly by substituting the relationships in \Cref{table:example_crossmoments} into \Cref{equation:variance,equation:correlation}.
\begin{equation}
\operatorname{Cov}(X_i, X_j) = \mathbb{E}[X_iX_j] - \mathbb{E}[X_i] \, \mathbb{E}[X_j] \qquad \forall i,j \in \{1,\dots,N\}
\label{equation:variance}
\end{equation}

\begin{equation}
\operatorname{Corr}(X_i, X_j) = \frac{\operatorname{Cov}(X_i,X_j)}{\sqrt{\operatorname{Var}(X_i)\operatorname{Var}(X_j)}} \qquad \forall i,j \in \{1,\dots,N\}
\label{equation:correlation}
\end{equation}

\clearpage
\paragraph{Example problem.}
Suppose the user specifies the following moments:
\begin{equation*}
\vec{\mu} = \left( 
\begin{array}{ccc}
0.3 & 0.5 & 0.6 \\
\end{array} \right) ^\mathsf{T}
\end{equation*}

\begin{equation*}
\mathbf{R} = \left( 
\begin{array}{ccc}
1 & 0.2 & 0.3\\
0.2& 1& 0.6\\
0.3& 0.6& 1
\end{array} \right)
\end{equation*}

One solution returned by the solver is:

\begin{equation*}
\vec{\alpha} = \left(
\begin{array}{r}
0.3213\\
0.0745\\
0.0260\\
0.2781\\
0.0257\\
0.0785\\
0.0270\\
0.1688
\end{array}
\right)
\end{equation*}

One can verify directly, by substituting these values into \Cref{equation:first_moment_example,equation:correlation}, that the stated means and correlations are recovered.

\paragraph{Example for simulating outcomes.}
Since $\vec{\alpha}$ defines the PMF over the states, inverse transform sampling can be applied.
Concretely, one draws a random number $u \sim U[0,1]$ and then computes $\inf\lbrace x: F(x)\geq u\rbrace$, as illustrated in the table below:

\begin{table*}[!h]
\centering
\begin{tabular}{rc}
\toprule
Generated number $u$ & Corresponding state  \\
{} & ($X_1 X_2 X_3$)\\
\midrule
$u < 0.3213$ & 000 \\
$0.3213 < u \leq 0.3958$ & 001 \\
$0.3958 < u \leq 0.4218$ & 010 \\
$0.4218 < u \leq 0.6999$ & 011 \\ 
$0.6999 < u \leq 0.7256$ & 100 \\
$0.7256 < u \leq 0.8041$ & 101 \\
$0.8041 < u \leq 0.8311$ & 110 \\
$0.8311 < u \leq 1.0000$ & 111 \\
\bottomrule
\end{tabular}
\end{table*}

\newpage
\section{Moment Identities and Reconstruction Formulas}\label{appendix:moment_identities}

This appendix collects the full identities linking the joint PMF and CDF of a Bernoulli vector to its cross-moments, together with proofs and several extensions.
In the main text we impose only the low-order constraints supplied by the user, namely means and pairwise correlations.
Here we record the complementary setting in which the entire cross-moment table is specified.

For a Bernoulli vector $X=(X_1,\ldots,X_N)\in\{0,1\}^N$, the joint PMF assigns a mass $p_{x_1,\ldots,x_N}$ to each of the $2^N$ outcomes in $\{0,1\}^N$.
Accordingly, a full specification of the law requires $2^N$ numbers subject to normalization.
The results below show how those atomic probabilities are reconstructed from the complete cross-moment table and when the resulting table is valid.

\subsection{Problem formulation}

We show here that the joint distribution of $N$ Bernoulli random variables is uniquely determined by its moments of order less than or equal to $N$.

Throughout this section, let $[N] := \{1,\ldots,N\}$.

We use the following notation throughout this section:
\begin{itemize}
    \item $p_i = \mathbb{P}(X_i = 1)$ denotes the marginal distribution of a Bernoulli random variable, with $p_i \in [0,1]$.
    \item $p_{x_1, \ldots, x_N} = \mathbb{P}(X_1 = x_1, \ldots, X_N=x_N)$, with $(x_1, \ldots, x_N) \in \{0,1\}^N$, denotes the joint PMF of the $N$ Bernoulli random variables. 
    \item $\gamma_i = \mathbb{E}[X_i] = p_i$ denotes the first order moments of a Bernoulli random variable. There are $\binom{N}{1} = N$ such moments.
    \item $\gamma_{i,j} = \mathbb{E}[X_i X_j]$ denotes the second-order cross-moment of a pair of Bernoulli random variables. There are $\binom{N}{2}$ such moments.
    \item $\gamma_{i,j,k} = \mathbb{E}[X_i X_j X_k]$ denotes the third-order cross-moment of a triplet of Bernoulli random variables. There are $\binom{N}{3}$ such moments.
    \item In general, $\gamma_{i_1,\ldots,i_k} = \mathbb{E}[X_{i_1} \cdots X_{i_k}]$ denotes the $k^{\rm th}$-order cross-moment of a $k$-tuple of Bernoulli random variables. There are $\binom{N}{k}$ such moments.
    \item For any subset $S=\{i_1,\ldots,i_k\}\subseteq [N]$, we write
    \[
        \gamma_S := \gamma_{i_1,\ldots,i_k} = \mathbb{E}\Big[\prod_{i\in S} X_i\Big],
    \]
    with the convention $\gamma_\varnothing := 1$.
    \item For any subset $S\subseteq [N]$, we write
    \[
        p_S := \mathbb{P}(X_i=1 \text{ iff } i\in S).
    \]
    This is the same quantity as $p_{x_1,\ldots,x_N}$ for the unique $x\in\{0,1\}^N$ such that $S=\{i:x_i=1\}$.
\end{itemize}
Throughout this section, we define $\Gamma$ as the collection of cross-moments
\[
    \Gamma := \{\gamma_S : S\subseteq [N]\},
\]
where $\gamma_\varnothing = 1$ is the zeroth-order moment condition.
Hence, $|\Gamma| = 2^N$.

\begin{lemma}\label{lemma:momCond}
Suppose that all elements of $\Gamma$ are given. The set of moment conditions up to order $N$ provides a total of
\begin{equation}
    \binom{N}{1} + \binom{N}{2} + \ldots + \binom{N}{N} = \sum_{k=0}^N \binom{N}{k} - 1 = 2^N - 1
\end{equation}
equations, and including the normalization constraint $\sum_{S\subseteq [N]} p_S = 1$ provides us with an additional equation leading to a total of $2^N$ equations. Henceforth, we refer to this set of $2^N$ equations as the \emph{moment equations} induced by cross-moments $\Gamma$:
\begin{equation}\label{eq:momEq}
    \gamma_S = \sum_{T\supseteq S} p_T,
\end{equation}
for all $S\subseteq [N]$. When $S=\varnothing$, we recover $\gamma_\varnothing = \sum_{T\subseteq [N]} p_T = 1$.
\end{lemma}

\begin{proof}
Fix any subset $S\subseteq [N]$. Since $X_i\in\{0,1\}$, the product $\prod_{i\in S} X_i$ equals $1$ exactly on those states whose set of ones contains $S$. Therefore,
\[
    \gamma_S
    = \mathbb{E}\Big[\prod_{i\in S} X_i\Big]
    = \sum_{T\subseteq [N]} p_T \prod_{i\in S} \mathbbm{1}\{i\in T\}
    = \sum_{T\supseteq S} p_T.
\]
For each $k\in\{1,\ldots,N\}$ there are $\binom{N}{k}$ subsets of $[N]$ of size $k$, hence $\binom{N}{k}$ moment equations of order $k$. Summing over $k=1,\ldots,N$ gives $2^N-1$ nontrivial equations, and adding the case $S=\varnothing$ yields the normalization equation. Altogether we obtain $2^N$ equations.
\end{proof}

\subsection{Solving the moment equations}

In the next lemma, we show that, given the set of all cross-moments $\Gamma$, the moment equations (\Cref{eq:momEq}) uniquely characterize the joint PMF of the $N$ Bernoulli random variables.

\begin{lemma}\label{lemma:momEq}
Given $\Gamma$, there exists a unique solution to the moment equations (\Cref{eq:momEq}). This solution uniquely characterizes the joint PMF of the $N$ Bernoulli random variables.
\end{lemma}

\begin{proof}
Write the unknown PMF values as a vector $p=(p_T)_{T\subseteq [N]}$ and the moments as $\gamma=(\gamma_S)_{S\subseteq [N]}$. Then \eqref{eq:momEq} can be written as
\[
    \gamma = M p,
\]
where the $2^N\times 2^N$ matrix $M$ is indexed by subsets of $[N]$ and has entries
\[
    M_{S,T} = \mathbbm{1}\{S\subseteq T\}.
\]
Order both the rows and the columns by increasing subset cardinality (breaking ties arbitrarily). If $M_{S,T}=1$, then $S\subseteq T$, hence $|S|\le |T|$. Therefore every nonzero entry lies on or above the diagonal, so $M$ is upper triangular. Moreover, $M_{S,S}=1$ for every $S$, and if $|S|=|T|$ with $S\subseteq T$, then necessarily $S=T$, so the diagonal entries are the only nonzero entries on each equal-cardinality block.

Thus $M$ is upper triangular with ones on the diagonal, and therefore
\[
    \det(M)=1\neq 0.
\]
Hence $M$ is invertible, so the system \eqref{eq:momEq} has a unique solution. Since the unknowns are exactly the $2^N$ PMF values, this unique solution uniquely characterizes the joint distribution.
\end{proof}

\subsection{Joint probability mass function}

We now provide a closed-form expression for the joint PMF of the $N$ Bernoulli random variables.

\begin{proposition}\label{prop:PMF}
Given the collection of cross-moments $\Gamma$, the joint PMF of $N$ Bernoulli random variables at $x = (x_1,\ldots,x_N) \in \{0,1\}^N$ is
\begin{equation}\label{eq:pmf_mobius}
    p_x \,=\, \sum_{T \subseteq Z(x)} (-1)^{|T|}\, \gamma_{A(x) \cup T},
\end{equation}
where $A(x) = \{ i : x_i = 1 \}$ and $Z(x) = \{ i : x_i = 0 \}$ form a partition of $\{1,\ldots,N\}$. Equivalently,
\begin{equation}
    p_x \,=\, \sum_{S \supseteq A(x)} (-1)^{|S|-|A(x)|} \, \gamma_S,
\end{equation}
which is the M\"obius inversion of the zeta transform on the Boolean lattice (see, e.g., \cite{bahadur1961}).
\end{proposition}

\begin{proof}
We have $\mathbbm{1}\{X=x\} = \prod_{i\in A(x)} X_i \prod_{j\in Z(x)} (1-X_j)$. Expanding the product over $Z(x)$ yields
\[
\prod_{j\in Z(x)} (1-X_j) = \sum_{T\subseteq Z(x)} (-1)^{|T|} \prod_{j\in T} X_j.
\]
Taking expectations gives
\[
 p_x = \mathbb{E}[\mathbbm{1}\{X=x\}] = \sum_{T\subseteq Z(x)} (-1)^{|T|} \, \mathbb{E}\big[\prod_{i\in A(x)} X_i \prod_{j\in T} X_j\big]
 = \sum_{T\subseteq Z(x)} (-1)^{|T|} \, \gamma_{A(x)\cup T}.
\]
This proves \eqref{eq:pmf_mobius}. The equivalent form obtained by indexing over $S = A(x)\cup T$ follows immediately. This is precisely the M\"obius inversion corresponding to the identities $\gamma_S = \sum_{x:\,A(x)\supseteq S} p_x$.
\end{proof}

\begin{corollary}[Validity of $\Gamma$ via nonnegativity of the implied PMF]\label{corr:PMFvalid}
Let $\Gamma=\{\gamma_S:S\subseteq [N]\}$ with $\gamma_\varnothing=1$. Then $\Gamma$ is realizable as the
moment collection of an $N$-dimensional Bernoulli vector if and only if the values $p_x$ defined by
\eqref{eq:pmf_mobius} satisfy $p_x\ge 0$ for all $x\in\{0,1\}^N$. In that case the resulting PMF is unique.
\end{corollary}

\begin{proof}
If $\Gamma$ comes from a Bernoulli vector $X$, then \Cref{prop:PMF} gives
\[
    p_x = \mathbb{P}(X=x)\ge 0
\]
for every $x$.

Conversely, suppose the values in \eqref{eq:pmf_mobius} are nonnegative. Reindex them as $(p_T)_{T\subseteq [N]}$. For any fixed $S\subseteq [N]$,
\begin{align*}
    \sum_{T\supseteq S} p_T
    &= \sum_{T\supseteq S}\ \sum_{U\supseteq T} (-1)^{|U|-|T|}\gamma_U \\
    &= \sum_{U\supseteq S} \gamma_U \sum_{T:\,S\subseteq T\subseteq U} (-1)^{|U|-|T|}.
\end{align*}
For a given $U\supseteq S$, the inner sum is
\[
    \sum_{T:\,S\subseteq T\subseteq U} (-1)^{|U|-|T|}
    = \sum_{V\subseteq U\setminus S} (-1)^{|V|}
    = (1-1)^{|U\setminus S|}
    = \mathbbm{1}\{U=S\}.
\]
Therefore $\sum_{T\supseteq S} p_T = \gamma_S$ for every $S\subseteq [N]$, so $p$ satisfies the moment equations. Taking $S=\varnothing$ gives
\[
    \sum_{T\subseteq [N]} p_T = \gamma_\varnothing = 1.
\]
Hence the nonnegative numbers $(p_T)$ form a PMF with moments $\Gamma$. Uniqueness follows from \Cref{lemma:momEq}.
\end{proof}

\paragraph{Remark on specification of $\Gamma$.}
The collection $\Gamma=\{\gamma_S:S\subseteq[N]\}$ contains all cross-moments up to order $N$, and hence
contains $2^N$ numbers (with $\gamma_\varnothing=1$ fixed). This full collection determines the joint law
uniquely. By contrast, in many applied settings one specifies only the marginals and a pairwise correlation
matrix. In general this pairwise information neither determines the full joint distribution nor fixes the
higher-order cross-moments. Thus, when only low-order information is prescribed, the main question is
whether there exists \emph{some} completion of the higher-order moments for which the implied atoms
$p_S$ are all nonnegative.

\subsection{Joint cumulative distribution function}

We now provide a closed-form expression for the joint CDF of the $N$ Bernoulli random variables.
For $t=(t_1,\dots,t_N)\in\mathbb{R}^N$, define
\[
F(t)\;:=\;\mathbb{P}(X_1\le t_1,\dots,X_N\le t_N).
\]
Note that, since $X_i\in\{0,1\}$, only the cases $t_i<0$, $0\le t_i<1$, and $t_i\ge 1$ matter.

\begin{proposition}[Closed-form joint CDF from cross-moments]\label{prop:CDF}
Let $t\in\mathbb{R}^N$ and define the index set
\[
C(t):=\{i\in[N]: t_i<1\}.
\]
If $t_i<0$ for some $i$, then $F(t)=0$. Otherwise,
\begin{equation}\label{eq:joint_cdf_gamma}
F(t)
\;=\;
\mathbb{E}\!\left[\prod_{i\in C(t)}(1-X_i)\right]
\;=\;
\sum_{B\subseteq C(t)}(-1)^{|B|}\,\gamma_B.
\end{equation}
Equivalently, in terms of the joint PMF $(p_S)$,
\begin{equation}\label{eq:joint_cdf_pmf}
F(t)
\;=\;
\sum_{A\subseteq [N]\setminus C(t)} p_A.
\end{equation}
In particular, for $x\in\{0,1\}^N$ with $S(x)=\{i:x_i=1\}$,
\[
F(x)
=
\mathbb{P}(X_i\le x_i\ \forall i)
=
\sum_{A\subseteq S(x)} p_A
=
\sum_{B\subseteq [N]\setminus S(x)} (-1)^{|B|}\,\gamma_B.
\]
\end{proposition}

\begin{proof}
If $t_i<0$ for some $i$, then $\{X_i\le t_i\}$ is empty because $X_i\ge 0$ almost surely, hence $F(t)=0$.

Assume now that $t_i\ge 0$ for all $i$. For each coordinate,
\[
\{X_i\le t_i\}
=
\begin{cases}
\{X_i=0\}, & 0\le t_i<1,\\
\Omega, & t_i\ge 1.
\end{cases}
\]
Hence
\[
F(t)=\mathbb{P}(X_i=0\ \forall i\in C(t)).
\]
Since $X_i\in\{0,1\}$, the indicator of the event $\{X_i=0\ \forall i\in C(t)\}$ is
$\prod_{i\in C(t)}(1-X_i)$, so
\[
F(t)=\mathbb{E}\!\left[\prod_{i\in C(t)}(1-X_i)\right].
\]
Expanding the product yields
\[
\prod_{i\in C(t)}(1-X_i)
=
\sum_{B\subseteq C(t)}(-1)^{|B|}\prod_{b\in B}X_b,
\]
and taking expectations gives \eqref{eq:joint_cdf_gamma}.

For \eqref{eq:joint_cdf_pmf}, note that $X_i=0$ for all $i\in C(t)$ is equivalent to saying that the set
of ones $\{i:X_i=1\}$ is a subset of $[N]\setminus C(t)$, hence the probability is the sum of the
corresponding atomic masses $p_A$.
The last displayed identities follow by taking $t=x\in\{0,1\}^N$, for which $C(x)=[N]\setminus S(x)$.
\end{proof}

\begin{corollary}[Validity of the CDF formula]\label{corr:CDFvalid}
Assume $\gamma_\varnothing=1$ and let $(p_A)$ denote the atomic probabilities recovered from
\Cref{prop:PMF}. Then the expression in
\Cref{prop:CDF} defines the joint CDF of a Bernoulli vector if and only if $p_A\ge 0$ for all
$A\subseteq[N]$.
\end{corollary}

\begin{proof}
If $p_A\ge 0$ for all $A$, then \Cref{corr:PMFvalid} shows $(p_A)$ is a valid PMF on $\{0,1\}^N$.
For this PMF, $F(t)=\mathbb{P}(X\le t)$ is a valid joint CDF by definition, and \Cref{prop:CDF}
shows that it equals \eqref{eq:joint_cdf_gamma}.

Conversely, if \eqref{eq:joint_cdf_gamma} is the joint CDF of some Bernoulli vector, then its atoms
$p_A=\mathbb{P}(X_i=1\text{ iff }i\in A)$ are nonnegative. By \Cref{prop:PMF}, these atoms are exactly the
values computed by \eqref{eq:pmf_mobius}.
\end{proof}

\subsection{Computational complexity}
\label{subsec:complexity}

The formulas above are closed-form, but the number of subsets of $[N]$ is $2^N$, so any algorithm that
takes the full collection $\Gamma=\{\gamma_A:A\subseteq[N]\}$ as input necessarily works with exponentially
many quantities in $N$.

Let
\[
K := |\Gamma| = 2^N.
\]
Then $K$ is the natural input size for the fully specified cross-moment representation, and
\[
N = \log_2 K.
\]

A direct implementation of \eqref{eq:pmf_mobius} computes all atoms separately. Since the total number of
pairs $(A,B)$ with $A\subseteq B$ is $3^N$, this naive approach has time complexity
\[
O(3^N).
\]

However, \eqref{eq:momEq} is a zeta transform on the Boolean lattice and \eqref{eq:pmf_mobius} is the
corresponding M\"obius inversion. Using the standard fast zeta/M\"obius transform on subsets, one can compute
all values $(p_A)_{A\subseteq[N]}$ in $O(N\,2^N)$ time and $O(2^N)$ memory.
Equivalently, in terms of the input size $K=2^N$, this is $O(K\log K)$ time and $O(K)$ memory.
Thus the dependence is exponential in the number of Bernoulli variables $N$, but nearly linear in the size
of the full input table $\Gamma$.

In \Cref{ssec:reduceMom} we explore ways to reduce the effective complexity when the goal is not to recover
the unique PMF associated with a fully specified moment table, but rather to obtain \emph{some} valid PMF
consistent with a prescribed collection of low-order moments.

\subsection{Geometry of first- and second-moment feasibility}\label{subsec:moment_polytope}
The PMF-LP can also be viewed geometrically. This viewpoint yields a support bound for feasible targets and
a separating certificate for infeasible ones.

\begin{proposition}[Convex-hull characterization, sparse realization, and infeasibility certificate]\label{prop:moment_polytope}
For each state $x\in\{0,1\}^N$, define the feature vector
\[
\phi(x):=\bigl(1,(x_i)_{i=1}^N,(x_ix_j)_{1\le i<j\le N}\bigr)\in\mathbb{R}^m,
\qquad m=1+N+\binom{N}{2},
\]
and define the target vector
\[
b:=\bigl(1,(\mu_i)_{i=1}^N,(\Sigma_{ij}+\mu_i\mu_j)_{1\le i<j\le N}\bigr).
\]
Then the requested moments are jointly Bernoulli-feasible if and only if
\[
b\in\operatorname{conv}\{\phi(x):x\in\{0,1\}^N\}.
\]
Whenever they are feasible, there is a compatible distribution supported on at most $m$ states.
Whenever they are infeasible, there are coefficients $c$, $(a_i)$, and $(a_{ij})$ such that
\[
q(x)=c+\sum_i a_i x_i+\sum_{i<j}a_{ij}x_ix_j
\]
satisfies $q(x)\ge0$ for every $x\in\{0,1\}^N$, but
\[
c+\sum_i a_i\mu_i+\sum_{i<j}a_{ij}(\Sigma_{ij}+\mu_i\mu_j)<0.
\]
Thus $q$ is a separating certificate: every Bernoulli law has $\mathbb{E}[q(X)]\ge0$, whereas the requested moments would force $\mathbb{E}[q(X)]<0$.
\end{proposition}

\begin{proof}
A PMF $(\alpha_x)$ satisfies the constraints exactly when
\[
b=\sum_{x\in\{0,1\}^N}\alpha_x\phi(x),
\qquad \alpha_x\ge0,\qquad \sum_x\alpha_x=1,
\]
which is precisely membership in the displayed convex hull. All feature vectors have first coordinate one,
so this convex hull lies in an affine space of dimension at most $m-1$. Carath\'eodory's theorem therefore
represents every feasible $b$ as a convex combination of at most $m$ feature vectors, proving the support
bound.

If $b$ is outside the finite convex hull, the strict separating-hyperplane theorem gives an affine functional
that is nonnegative on every $\phi(x)$ and negative at $b$. Absorbing its intercept into the constant
coordinate produces coefficients $c,a_i,a_{ij}$ with the stated properties. Conversely, such coefficients
rule out feasibility because taking expectations of the pointwise inequality $q(X)\ge0$ contradicts the
value forced by the target moments.
\end{proof}

\subsection{Reducing the number of moments}\label{ssec:reduceMom}

The results above show that the full cross-moment table $\Gamma=\{\gamma_A:A\subseteq[N]\}$ uniquely
determines the joint law, and that validity of $\Gamma$ is equivalent to nonnegativity of the implied atomic
probabilities $(p_A)_{A\subseteq[N]}$ (\Cref{corr:PMFvalid}). While the full specification is exact, it is
often unnecessary in practice: one typically needs \emph{some} valid joint PMF that matches a limited
collection of low-order constraints (e.g.\ means and pairwise correlations), not the unique PMF corresponding
to a fully prescribed $\Gamma$.

The key observation is that, by \Cref{prop:PMF}, each atom $p_A$ is an explicit \emph{linear} function of the
moments. Consequently, imposing $p_A\ge 0$ for all $A$ yields a family of linear inequalities in the unknown
moments. This makes it natural to search for a \emph{low-dimensional} or \emph{sparse} moment table
$\Gamma$ that is consistent with the user-specified low-order moments and satisfies the nonnegativity
constraints.

\paragraph{A truncation principle.}
A particularly simple way to reduce the number of moments is to enforce that all sufficiently high-order
cross-moments vanish. Intuitively, this forbids large sets of simultaneous successes and forces the PMF to
live on a small part of $\{0,1\}^N$.

\begin{lemma}[Moment truncation and support truncation]\label{lem:moment_truncation_support}
Fix an integer $k\in\{0,1,\dots,N\}$. Suppose a moment table $(\gamma_A)_{A\subseteq[N]}$ satisfies
\[
\gamma_A = 0 \qquad \text{for all } A\subseteq[N]\text{ with }|A|>k,
\qquad\text{and}\qquad \gamma_\varnothing=1.
\]
Define the atomic probabilities $(p_A)_{A\subseteq[N]}$ using the equivalent subset form in
\Cref{prop:PMF}. Then
\[
p_A = 0 \qquad \text{for all }A\subseteq[N]\text{ with }|A|>k,
\]
and for all $A\subseteq[N]$ with $|A|\le k$,
\begin{equation}\label{eq:truncated_mobius}
p_A
\;=\;
\sum_{\substack{B\supseteq A\\ |B|\le k}}
(-1)^{|B|-|A|}\,\gamma_B.
\end{equation}
In particular, if the values computed in \eqref{eq:truncated_mobius} satisfy $p_A\ge 0$ for all $A$, then
$(p_A)$ is a valid joint PMF supported on $\{x\in\{0,1\}^N:\sum_{i=1}^N x_i\le k\}$.
Conversely, for any Bernoulli vector, support on this set implies $\gamma_A=0$ for every $|A|>k$.
\end{lemma}

\begin{proof}
If $|A|>k$ and $B\supseteq A$, then $|B|\ge |A|>k$, so $\gamma_B=0$ by assumption, and hence $p_A=0$.
If $|A|\le k$, then only supersets $B\supseteq A$ with $|B|\le k$ can contribute nonzero terms, yielding
\eqref{eq:truncated_mobius}. The validity claim follows from \Cref{corr:PMFvalid}: nonnegativity of all atoms
is equivalent to existence of a Bernoulli vector with the prescribed moments, and the previous argument
shows the resulting PMF has no mass on atoms with more than $k$ ones.

For the converse, suppose the law is supported on states with at most $k$ ones. If $|A|>k$, then on every
supported state at least one coordinate in $A$ is zero. Hence $\prod_{i\in A}X_i=0$ almost surely and
$\gamma_A=0$.
\end{proof}

\paragraph{Complexity under truncation.}
If $\gamma_A=0$ for $|A|>k$, then the number of potentially nonzero moments is
\[
M_k := \sum_{m=0}^{k}\binom{N}{m},
\]
which is polynomial in $N$ for fixed $k$. Moreover, for each fixed $A$ with $|A|\le k$, the sum
\eqref{eq:truncated_mobius} ranges over supersets of $A$ of size at most $k$, whose number is at most
\[
\sum_{r=0}^{k-|A|}\binom{N-|A|}{r}\le M_k.
\]
Thus, when $k$ is fixed and small, evaluating the nonzero atoms and checking $p_A\ge 0$ can be done in
time and memory scaling polynomially in $N$.

\paragraph{The order-$2$ (pairwise-only) truncation as a ``quick'' construction.}
The most aggressive nontrivial truncation is $k=2$, i.e.\ we set all moments of order $\ge 3$ to zero.
Then the PMF assigns mass only to outcomes with at most two ones. In this case the inversion simplifies to
closed forms that depend only on the singleton and pair moments.

\begin{corollary}[Explicit PMF under pairwise-only moments]\label{corr:pairwise_only_pmf}
Assume $\gamma_\varnothing=1$ and $\gamma_A=0$ for all $|A|\ge 3$. Then the implied PMF $(p_A)$ is given by
\[
p_{\{i,j\}} = \gamma_{\{i,j\}}
\quad (i\ne j),\qquad
p_{\{i\}} = \gamma_{\{i\}} - \sum_{j\in[N]\setminus\{i\}} \gamma_{\{i,j\}},
\qquad
p_{\varnothing} = 1 - \sum_{i\in[N]}\gamma_{\{i\}} + \sum_{1\le i<j\le N}\gamma_{\{i,j\}},
\]
and $p_A=0$ for all $|A|\ge 3$. In particular, \emph{within this order-$2$ truncation family}, the above
formulas define a valid joint distribution if and only if
\[
\gamma_{\{i,j\}}\ge 0 \ \ (\forall i\ne j),\qquad
\gamma_{\{i\}} - \sum_{j\ne i}\gamma_{\{i,j\}} \ge 0 \ \ (\forall i),\qquad
1 - \sum_i\gamma_{\{i\}} + \sum_{i<j}\gamma_{\{i,j\}} \ge 0.
\]
\end{corollary}

\begin{proof}
For $A=\{i,j\}$, the only superset $B\supseteq A$ with $|B|\le 2$ is $B=A$, so $p_{\{i,j\}}=\gamma_{\{i,j\}}$.
For $A=\{i\}$, the contributing supersets with $|B|\le 2$ are $B=\{i\}$ and $B=\{i,j\}$ with $j\ne i$,
giving $p_{\{i\}}=\gamma_{\{i\}}-\sum_{j\ne i}\gamma_{\{i,j\}}$.
For $A=\varnothing$, the contributing sets are $B=\varnothing$, all singletons, and all pairs, yielding the
displayed formula for $p_\varnothing$. The nonnegativity conditions are exactly $p_A\ge 0$ for all atoms
that can be nonzero under the order-$2$ truncation.
\end{proof}

\paragraph{Matching means and correlations.}
If one is given target means $\mu_i=\mathbb{E}[X_i]$ and Pearson correlations $\rho_{ij}=\operatorname{corr}(X_i,X_j)$,
then the corresponding pair moments are
\[
\gamma_{\{i\}} = \mu_i,
\qquad
\gamma_{\{i,j\}}
=
\mathbb{E}[X_iX_j]
=
\mu_i\mu_j + \rho_{ij}\sqrt{\mu_i(1-\mu_i)\,\mu_j(1-\mu_j)}.
\]
Under the pairwise-only truncation, one computes $(p_A)$ via \Cref{corr:pairwise_only_pmf} and obtains a
valid joint PMF exactly when the resulting atoms are nonnegative.

\paragraph{Important interpretation.}
The order-$2$ truncation is a \emph{restricted model class}: it forces the support to lie on outcomes with at
most two ones. Therefore, failure of the inequalities in \Cref{corr:pairwise_only_pmf} does \emph{not}
imply that the specified means and pairwise correlations are globally infeasible for Bernoulli variables. It
only shows that they cannot be realized \emph{within this truncated-support family}. If these inequalities
fail, a natural remedy is to increase $k$ and allow some higher-order moments to be nonzero, while still
keeping $k$ small enough that $M_k$ remains tractable.

\paragraph{Feasibility viewpoint.}
More generally, one may fix a family of ``allowed'' moments (for instance, all $A$ with $|A|\le k$) and
treat the remaining moments as zero. By \Cref{corr:PMFvalid}, the problem of finding \emph{any} valid
distribution consistent with prescribed low-order information reduces to finding a moment table for which
the implied atoms are nonnegative. Under truncation (or other sparsity patterns), this becomes a linear
feasibility problem over a number of variables scaling like $M_k$ rather than $2^N$.

\paragraph{A concrete LP template.}
Suppose the user specifies means $\mu_i=\mathbb{E}[X_i]$ and pairwise correlations $\rho_{ij}$, and fix a
truncation order $k\ge 2$. Then
\[
\gamma_\varnothing=1,\qquad
\gamma_{\{i\}}=\mu_i,\qquad
\gamma_{\{i,j\}}=\mu_i\mu_j+\rho_{ij}\sqrt{\mu_i(1-\mu_i)\,\mu_j(1-\mu_j)}
\]
are fixed, the moments $\gamma_A$ with $3\le |A|\le k$ are decision variables, and the moments with $|A|>k$
are set to zero. The resulting order-$k$ completion problem is the LP
\[
\text{find } (\gamma_A)_{A\subseteq[N],\,3\le |A|\le k}
\]
subject to
\[
\sum_{\substack{B\supseteq A\\ |B|\le k}}(-1)^{|B|-|A|}\gamma_B \;\ge\; 0
\qquad \text{for every } A\subseteq[N]\text{ with }|A|\le k,
\]
where the singleton and pair moments are held fixed at the target values above. A zero objective gives a
pure feasibility problem. In practice, one may add a simple tie-breaking linear objective, such as minimizing
$\sum_{3\le |A|\le k}\gamma_A$ to favor sparse higher-order structure or maximizing $p_\varnothing$ to place
more mass on the all-zero state. If the LP is infeasible, then either the target lies outside the order-$k$
truncated family or $k$ is too small to capture the required dependence.

\subsection{Generating correlated Bernoulli random variables}

Once a valid joint law has been specified, either through the full table $\Gamma$ or through a reduced choice as in \Cref{ssec:reduceMom}, the remaining task is to generate samples $(X_1,\dots,X_N)$.

\paragraph{Direct sampling from the PMF.}
If one explicitly computes and stores all atoms $(p_A)_{A\subseteq[N]}$, then one can sample an atom
$A$ with probability $p_A$ and output $x=\mathbf{1}_A$. This is straightforward but requires $O(2^N)$
storage in general. Under a support truncation $|A|\le k$, the number of nonzero atoms is at most
$M_k=\sum_{m=0}^k\binom{N}{m}$, which can be feasible for small $k$.

\paragraph{Cylinder (partial-assignment) probabilities.}
For disjoint index sets $O,Z\subseteq[N]$, define the cylinder probability
\[
w(O,Z)\;:=\;\mathbb{P}\!\big(X_i=1\ \forall i\in O,\ X_j=0\ \forall j\in Z\big).
\]
This quantity can be computed (i) from an explicitly stored PMF by summing all atoms whose set of ones
$A$ satisfies $O\subseteq A\subseteq [N]\setminus Z$, (ii) from a table of CDF values on $\{0,1\}^N$ by finite
differences, or (iii) directly from cross-moments via inclusion--exclusion as below.

\paragraph{On-the-fly sampling via sequential conditioning.}
To avoid precomputing the full PMF, one may sample coordinates sequentially using conditional probabilities computed from cross-moments.
The following identity extends the expansion used in the proof of \Cref{prop:PMF} and specializes to \Cref{prop:CDF} when no ones are fixed.

\begin{lemma}[Probabilities of partial assignments from cross-moments]\label{lem:partial_assignment_prob}
Let $O,Z\subseteq[N]$ be disjoint index sets, interpreted as constraints
\[
X_i=1 \text{ for } i\in O,
\qquad
X_j=0 \text{ for } j\in Z,
\]
with no restriction on coordinates in $[N]\setminus (O\cup Z)$.
Then
\begin{equation}\label{eq:partial_event_prob}
w(O,Z)
\;=\;
\mathbb{P}\!\big(X_i=1\ \forall i\in O,\ X_j=0\ \forall j\in Z\big)
\;=\;
\mathbb{E}\!\left[\prod_{i\in O}X_i\prod_{j\in Z}(1-X_j)\right]
\;=\;
\sum_{B\subseteq Z}(-1)^{|B|}\,\gamma_{O\cup B}.
\end{equation}
\end{lemma}

\begin{proof}
Since $X_i\in\{0,1\}$, the product $\prod_{i\in O}X_i\prod_{j\in Z}(1-X_j)$ is the indicator of the event
$\{X_i=1\ \forall i\in O,\ X_j=0\ \forall j\in Z\}$, hence the first equality.
Expanding $\prod_{j\in Z}(1-X_j)$ and taking expectations gives the sum over $B\subseteq Z$ with
$\mathbb{E}[\prod_{i\in O\cup B}X_i]=\gamma_{O\cup B}$.
\end{proof}

Fix an ordering $\pi(1),\dots,\pi(N)$ of the indices.
Suppose we have already sampled $(X_{\pi(1)},\dots,X_{\pi(m)})$ and define the sets of realized ones and zeros
\[
O_m := \{\pi(\ell): 1\le \ell\le m,\ X_{\pi(\ell)}=1\},
\qquad
Z_m := \{\pi(\ell): 1\le \ell\le m,\ X_{\pi(\ell)}=0\}.
\]
Then the conditional probability that the next variable equals $1$ can be written compactly as the ratio
\begin{equation}\label{eq:sequential_conditional_ratio}
\mathbb{P}\!\big(X_{\pi(m+1)}=1 \,\big|\, X_i=1\ \forall i\in O_m,\ X_j=0\ \forall j\in Z_m\big)
=
\frac{w(O_m\cup\{\pi(m+1)\},Z_m)}{w(O_m,Z_m)},
\end{equation}
whenever $w(O_m,Z_m)>0$. Using \Cref{lem:partial_assignment_prob} to expand numerator and denominator
yields the explicit cross-moment form
\begin{equation}\label{eq:sequential_conditional}
\mathbb{P}\!\big(X_{\pi(m+1)}=1 \,\big|\, X_i=1\ \forall i\in O_m,\ X_j=0\ \forall j\in Z_m\big)
=
\frac{
\sum_{B\subseteq Z_m}(-1)^{|B|}\,\gamma_{O_m\cup\{\pi(m+1)\}\cup B}
}{
\sum_{B\subseteq Z_m}(-1)^{|B|}\,\gamma_{O_m\cup B}
}.
\end{equation}
Sampling proceeds by drawing $U_{m+1}\sim \operatorname{Unif}(0,1)$ and setting $X_{\pi(m+1)}=1$ iff
$U_{m+1}$ is less than the right-hand side of \eqref{eq:sequential_conditional}. Repeating for
$m=0,1,\dots,N-1$ yields a full draw $(X_1,\dots,X_N)$ with the desired joint law.

\paragraph{Implementation details.}
To implement \eqref{eq:sequential_conditional}, one first fixes an ordering $\pi$ and stores the fitted truncated moment table $\Gamma_k$.
At each step, the state of the recursion is fully described by the pair $(O_m,Z_m)$, so it is natural to memoize the cylinder masses $w(O,Z)$ and reuse them whenever the same partial assignment reappears.
Under order-$k$ truncation, only terms with $|O|+|B|\le k$ can contribute, so each evaluation of $w(O,Z)$ reduces to a sum over a subset family of size at most $M_k$.
In floating-point implementations, tiny negative values caused by solver tolerances should be clipped to zero before forming the ratio in \eqref{eq:sequential_conditional}.
If the denominator is numerically zero, one can either declare the prefix impossible and backtrack or, in a simpler implementation, fall back to a direct PMF sampler if the fitted atoms have already been materialized.
In this sense, sequential conditioning is most useful when one wants to work from the reduced moment table $\Gamma_k$ without explicitly storing all nonzero atoms.

\paragraph{Computational cost and the role of sparsity.}
Formula \eqref{eq:sequential_conditional} requires evaluating sums over $B\subseteq Z_m$, which is
exponential in $m$ in the worst case. Interpreted via \eqref{eq:sequential_conditional_ratio}, this means
that exact sequential sampling requires an exact oracle for $w(O,Z)$, which in general is as hard as
recovering the full joint distribution (and thus still requires, explicitly or implicitly, accounting for all
$2^N$ states).

However, under the truncation scheme of \Cref{lem:moment_truncation_support} with $\gamma_A=0$ for $|A|>k$,
only terms with $|O_m|+|B|\le k$ can contribute. Hence, for fixed $k$, each of the sums in
\eqref{eq:sequential_conditional} contains at most
\[
\sum_{r=0}^{k-|O_m|} \binom{|Z_m|}{r}
\;\le\;
\sum_{r=0}^{k} \binom{N}{r}
\;=\;
M_k
\]
terms, which is polynomial in $N$ for fixed $k$. In this regime, one can generate a full vector
$(X_1,\dots,X_N)$ in time on the order of $N\,M_k$ without storing the full $2^N$-atom PMF.

\Cref{alg:truncated_sequential_sampler} summarizes one way to operationalize the truncation and
sequential-conditioning ideas above.

\begin{figure}[t]
\centering
\refstepcounter{algorithm}\label{alg:truncated_sequential_sampler}
\fbox{%
\begin{minipage}{0.95\linewidth}
\small
\textbf{Algorithm~\thealgorithm. Truncated Moment Completion with Sequential Conditioning}

\textbf{Input:} target means $(\mu_i)_{i=1}^N$, target pairwise correlations $(\rho_{ij})_{1\le i<j\le N}$,
truncation order $k$, number of samples $L$, and a variable order $\pi(1),\dots,\pi(N)$.

\textbf{Offline completion step}
\begin{enumerate}[itemsep=0.25ex, topsep=0.4ex, leftmargin=2.4em]
\item Set $\gamma_\varnothing=1$, $\gamma_{\{i\}}=\mu_i$, and
\[
\gamma_{\{i,j\}}=\mu_i\mu_j+\rho_{ij}\sqrt{\mu_i(1-\mu_i)\,\mu_j(1-\mu_j)}
\qquad (1\le i<j\le N).
\]
\item Introduce decision variables $\gamma_A$ for all $A\subseteq[N]$ with $3\le |A|\le k$, and set
$\gamma_A=0$ for all $|A|>k$.
\item Solve the linear feasibility problem
\[
p_A=\sum_{\substack{B\supseteq A\\ |B|\le k}}(-1)^{|B|-|A|}\gamma_B \;\ge\; 0
\qquad \text{for all } A\subseteq[N]\text{ with }|A|\le k,
\]
together with any chosen bounds on the unknown higher-order moments.
\item If the LP is infeasible, either stop or increase $k$. Otherwise store the fitted truncated moment table
$\Gamma_k=\{\gamma_A: |A|\le k\}$.
\end{enumerate}

\textbf{Online sampling step}
\begin{enumerate}[itemsep=0.25ex, topsep=0.4ex, leftmargin=2.4em]
\item For each sample $\ell=1,\dots,L$, initialize $O_0=Z_0=\varnothing$.
\item For $m=0,1,\dots,N-1$, compute
\[
w(O_m,Z_m)=\sum_{\substack{B\subseteq Z_m\\ |O_m|+|B|\le k}}(-1)^{|B|}\gamma_{O_m\cup B}
\]
and
\[
q_m=
\frac{
\sum_{\substack{B\subseteq Z_m\\ |O_m|+1+|B|\le k}}(-1)^{|B|}\gamma_{O_m\cup\{\pi(m+1)\}\cup B}
}{
\sum_{\substack{B\subseteq Z_m\\ |O_m|+|B|\le k}}(-1)^{|B|}\gamma_{O_m\cup B}
}.
\]
\item Draw $U_m\sim \operatorname{Unif}(0,1)$ and set $X_{\pi(m+1)}=1$ if $U_m\le q_m$, otherwise set
$X_{\pi(m+1)}=0$.
\item Update
\[
O_{m+1}=O_m\cup\{\pi(m+1):X_{\pi(m+1)}=1\},\qquad
Z_{m+1}=Z_m\cup\{\pi(m+1):X_{\pi(m+1)}=0\},
\]
and continue until all coordinates are assigned.
\end{enumerate}

\textbf{Output:} draws $X^{(1)},\dots,X^{(L)}\in\{0,1\}^N$ from the fitted order-$k$ truncated model.
\end{minipage}%
}
\end{figure}

\paragraph{Connection with the joint CDF.}
The special case $O=\varnothing$ in \eqref{eq:partial_event_prob} gives
\[
\mathbb{P}(X_j=0\ \forall j\in Z)=\sum_{B\subseteq Z}(-1)^{|B|}\gamma_B,
\]
which is exactly the joint CDF formula \eqref{eq:joint_cdf_gamma} evaluated at any $t$ with $C(t)=Z$.
Thus, the same inclusion--exclusion structure that yields the closed-form CDF also provides the event
probabilities required for sequential (tree-based) inverse-transform sampling.

\paragraph{Extending beyond pairwise truncation.}
\begin{proposition}[Limits of order-2 truncation]\label{prop:order3_limitation}
There exist first- and second-order moments (means and pairwise correlations) for which no choice of
third-order moments yields a valid joint Bernoulli distribution. In other words, an infeasible pairwise
truncation cannot always be rescued by adding $\gamma_{i,j,k}$ variables. In fact, even consistent pairwise
marginals may have no joint extension \cite{Vorobev1962,PadmanabhanNatarajan2021}, so there is no guarantee
that specifying order-3 moments will make $p_A\ge 0$ for all atoms.
\end{proposition}

\begin{proof}
Consider the $N=3$ specification
\[
\gamma_{\{1\}}=\gamma_{\{2\}}=\gamma_{\{3\}}=\frac12,
\qquad
\gamma_{\{i,j\}}=0 \quad \text{for all } i\neq j.
\]
Each requested pair is individually valid: it requires two Bernoulli variables with mean $1/2$ to be
perfectly negatively correlated. However, three variables cannot be pairwise perfect complements.
Algebraically, let $t:=\gamma_{\{1,2,3\}}=p_{\{1,2,3\}}$. M\"obius inversion gives
\[
p_{\{i,j\}}=\gamma_{\{i,j\}}-t=-t.
\]
Any valid law requires both $t=p_{\{1,2,3\}}\ge0$ and $p_{\{i,j\}}\ge0$, so necessarily $t=0$.
But then
\[
p_\varnothing
=1-\sum_i\gamma_{\{i\}}+\sum_{i<j}\gamma_{\{i,j\}}-t
=-\frac12<0.
\]
Thus no choice of third-order moment repairs the pairwise table. This directly demonstrates why pairwise
Bernoulli feasibility is not sufficient for joint feasibility.
\end{proof}

\paragraph{Required order of moments versus $N$.}
\begin{proposition}[High-order truncation needed]\label{prop:order_k_needed}
Let $\gamma_{\{i\}}$ be prescribed marginal means. If all moments above order $k$ are truncated to zero, then
any valid joint law has support contained in $\{x\in\{0,1\}^N: \sum_{i=1}^N x_i\le k\}$ by
\Cref{lem:moment_truncation_support}. Consequently,
\[
\mathbb{E}\!\left[\sum_{i=1}^N X_i\right]
=
\sum_{i=1}^N \gamma_{\{i\}}
\le k.
\]
Therefore, whenever $\sum_i \gamma_{\{i\}}$ is of order $N$, one must take $k=\Omega(N)$. In particular,
a truncation order $k=O(\log N)$ cannot guarantee feasibility in general. Moreover, if the means
$\gamma_{\{i\}}$ are sampled independently from a distribution on $[0,1]$ with positive mean $\mu>0$, then
Hoeffding's inequality implies
\[
\sum_{i=1}^N \gamma_{\{i\}} = \mu N + o(N)
\quad\text{with high probability,}
\]
so feasibility requires $k=\Omega(N)$ with high probability \cite{Hoeffding1963}.
\end{proposition}

\begin{proof}
By \Cref{lem:moment_truncation_support}, truncating above order $k$ forces $p_A=0$ for all $|A|>k$.
Hence
\[
\sum_{i=1}^N X_i \le k
\qquad \text{almost surely.}
\]
Taking expectations gives
\[
\sum_{i=1}^N \gamma_{\{i\}}
=
\mathbb{E}\!\left[\sum_{i=1}^N X_i\right]
\le k.
\]
Thus, if the prescribed means satisfy $\sum_i \gamma_{\{i\}}>k$, no distribution within the order-$k$
truncation family can realize them.

Now suppose $\gamma_{\{1\}},\dots,\gamma_{\{N\}}$ are i.i.d.\ random variables in $[0,1]$ with
$\mathbb{E}[\gamma_{\{i\}}]=\mu>0$. Hoeffding's inequality gives exponential concentration of
\[
\sum_i \gamma_{\{i\}}
\]
around $\mu N$ \cite{Hoeffding1963}. Therefore
\[
\sum_{i=1}^N \gamma_{\{i\}} = \Theta(N)
\qquad \text{with high probability.}
\]
Combining this with the necessary condition $\sum_i\gamma_{\{i\}}\le k$ yields $k=\Omega(N)$ with high
probability. In particular, $k=O(\log N)$ is insufficient in general.
\end{proof}

\newpage
\section{Sparse-Support Refinement of the PMF-LP}\label{sec:pmf_lp_refinement}
The full PMF-LP assigns one variable to each state in $\{0,1\}^N$, which is its main computational bottleneck.
At the same time, the number of independent moment constraints is only
\begin{equation*}
m \;=\; 1 + N + \binom{N}{2}
\end{equation*}
This gap motivates a sparse-support strategy: rather than optimizing over all $2^N$ states at once, one works with a small candidate support and enlarges it only when necessary.

Formally, let $\mathcal{S}_t \subset \{0,1\}^N$ be the current working set.
We solve a restricted master LP with weights $w_x$ for $x\in\mathcal{S}_t$:
\begin{subequations}
\begin{align}
\min_{w,\xi^+,\xi^-}\quad & \mathbf{1}^\mathsf{T}(\xi^+ + \xi^-)\\
\text{s.t.}\quad
& \sum_{x\in\mathcal{S}_t} w_x = 1,\\
& \sum_{x\in\mathcal{S}_t} w_x x_i + \xi_i^+ - \xi_i^- = \mu_i,\qquad i=1,\ldots,N,\\
& \sum_{x\in\mathcal{S}_t} w_x x_i x_j + \xi_{ij}^+ - \xi_{ij}^- = \Sigma_{ij}+\mu_i\mu_j,\qquad 1\le i<j\le N,\\
& w,\xi^+,\xi^- \ge 0.
\end{align}
\end{subequations}
If the optimal slack is zero, then the current support already delivers an exact moment match.
More strongly, \Cref{prop:moment_polytope} proves that every feasible target admits an exact law on at most $m$ states, even though identifying such a support may remain computationally difficult.

The remaining question is how to add states.
Using dual multipliers from the restricted master, one can define a pricing score of the form
\[
q(x) \;=\; \lambda_0 + \sum_i \lambda_i x_i + \sum_{i<j}\lambda_{ij}x_ix_j,\qquad x\in\{0,1\}^N,
\]
and seek a state with large improvement.
This is a binary quadratic optimization subproblem, which can be handled approximately (e.g., local search, simulated annealing, or QUBO-style heuristics), then fed back into the master LP.

At a high level, the procedure is:
\begin{enumerate}[itemsep=-1ex]
\item Initialize $\mathcal{S}_0$ with random and/or structured states.
\item Solve the restricted master LP.
\item Generate one or more new states via the pricing subproblem.
\item Add them to $\mathcal{S}_t$ and repeat until slack is zero or progress stalls.
\end{enumerate}
This refinement does not remove worst-case exponential behavior, but it can reduce memory use and runtime substantially on structured instances.

A related idea is to sample a subset of variables first (for example $X_1,\ldots,X_k$) and then solve an LP for the remainder.
In general, however, that decomposition is not exact unless one enforces conditional moment constraints for each sampled prefix, which reintroduces exponential complexity.
The sparse-support working-set approach avoids that issue by keeping all constraints global.

\end{document}